\documentclass[a4paper,onecolumn,11pt,accepted=2026-01-28]{quantumarticle}
\pdfoutput=1
\usepackage[utf8]{inputenc}
\usepackage[british]{babel}
\usepackage[T1]{fontenc}
\usepackage{amsmath}
\usepackage{hyperref}
\usepackage{cleveref}
\usepackage{microtype}

\usepackage{tikz}
\usepackage{lipsum}

\usepackage[numbers,sort&compress]{natbib}

\usepackage{graphicx}
 \usepackage{multirow}
\usepackage[table,xcdraw]{xcolor}
\usepackage{float}  
\usepackage{booktabs}
\usepackage{tikz-cd}
\usepackage{tikz}
\usetikzlibrary{decorations.markings,positioning}
\usetikzlibrary{shadows.blur}
\usepackage[export]{adjustbox}
\usepackage{standalone}
\usepackage[font=small]{caption}
\usepackage{subcaption}
\usepackage{tabularray}
\usepackage{xspace}

\usepackage{enumitem}

\usepackage{amsthm}
\usepackage{xfrac}
\usepackage{braket}
\usepackage{physics}

\usepackage{cleveref}

\newcommand{\ie}{\text{i.e.\@}\xspace}
\newcommand{\eg}{\text{e.g.\@}\xspace}

\newcommand{\textok}{\textsf{ok}}
\newcommand{\textfail}{\textsf{fail}}
\newcommand{\textyes}{\textsf{yes}}
\newcommand{\textno}{\textsf{no}}

\theoremstyle{plain}
\newtheorem{theorem}{Theorem}

\theoremstyle{definition}

\theoremstyle{remark}
\newtheorem{remark}{Remark}

\theoremstyle{definition}
\newtheorem{assumption}{Assumption}

\newtheorem{assumptionalt}{Assumption}[assumption]
\crefname{assumptionalt}{assumption}{assumptions}
\Crefname{assumptionalt}{Assumption}{Assumptions}
\crefrangeformat{assumptionalt}{assumptions~#3#1#4--#5#2#6}
\Crefrangeformat{assumptionalt}{Assumptions~#3#1#4--#5#2#6}

\newenvironment{assumptionp}[1]{
  \renewcommand\theassumptionalt{#1}
  \assumptionalt
}{\endassumptionalt}

\newcommand{\AOE}{\hyperlink{AOE}{\color{darkgray}{\textsc{AOE}}}}
\newcommand{\AOEfull}{\hyperlink{AOE}{\color{darkgray}{{\textsc{Absoluteness of Observed Events}}}}}
\newcommand{\possBorn}{\hyperlink{poss_Born}{\color{darkgray}{\textsc{Born Validity}}}}
\newcommand{\BornCAOE}{\hyperlink{Born_Compat_AOE}{\color{darkgray}{\textsc{Born Compatibility}}}}
\newcommand{\BornCPersK}{\hyperlink{Born_Compat_PersK}{\color{darkgray}{\textsc{Born Compatibility}}}}
\newcommand{\universality}{\hyperlink{Universality}{\color{darkgray}{\textsc{Universality of Quantum Theory}}}}
\newcommand{\PersK}{\hyperlink{Pers_K}{\color{darkgray}{\textsc{Personal Knowledge}}}}
\newcommand{\ClassicalT}{\hyperlink{Classical_T}{\color{darkgray}{\textsc{Classical Agreement}}}}
\newcommand{\ROE}{\hyperlink{ROE}{\color{darkgray}{\textsc{Relativity of Observed Events}}}}
\newcommand{\BornP}{\hyperlink{BornP}{\color{darkgray}{\textsc{Born Practicality}}}}

\usepackage{etoolbox}

\makeatletter
\def\@printauthorextrainfo#1{
	\csname @authorname#1\endcsname\vphantom{gy}:
	\ifcsdef{author#1emails}
	{%
		\forlistcsloop{\@@spaceafter}{author#1emails}%
	}
	{}%
	\ifcsdef{author#1homepages}
	{%
		\ifcsdef{author#1emails}
		{\unskip, }
		{}%
		\forlistcsloop{\@@spaceafter}{author#1homepages}%
	}
	{}%
	\ifcsdef{author#1thanks}
	{%
		\ifcsdef{author#1emails}
		{%
			\ifcsdef{author#1homepages}
			{\unskip, }
			{\unskip, }%
		}
		{%
			\ifcsdef{author#1homepages}
			{\unskip, }
			{}%
		}%
		\forlistcsloop{\@@spaceafter}{author#1thanks}%
	}
	{}%
}
\makeatother

\begin{document}

\title{A refined Frauchiger--Renner paradox based on strong contextuality}

\author{Laurens Walleghem}
\orcid{0000-0003-4459-1191}
\email{laurens.walleghem@york.ac.uk}
\affiliation{Department of Mathematics, University of York, York, United Kingdom}
\affiliation{INL -- International Iberian Nanotechnology Laboratory, Braga, Portugal}

\author{Rui Soares Barbosa}
\orcid{0000-0002-0465-8518}
\email{rui.soaresbarbosa@inl.int}
\affiliation{INL -- International Iberian Nanotechnology Laboratory, Braga, Portugal}

\author{Matthew F. Pusey}
\orcid{0000-0002-6189-7144}
\email{m@physics.org}
\affiliation{Department of Mathematics, University of York, York, United Kingdom}

\author{Stefan Weigert}
\orcid{0000-0002-6647-3252}
\email{stefan.weigert@york.ac.uk}
\affiliation{Department of Mathematics, University of York, York, United Kingdom}

\maketitle

\begin{abstract}
 The Frauchiger--Renner paradox derives an inconsistency when quantum theory is used to describe the use of itself,
    by means of a scenario where agents model other agents quantumly and reason about each other's knowledge.
    We observe that logical contextuality (\`a la Hardy) is the key ingredient of the FR paradox, and we provide a stronger paradox based on the strongly contextual GHZ--Mermin scenario.
    In contrast to the FR paradox, this GHZ--FR paradox neither requires post-selection nor any reasoning by observers who are modelled quantumly. 
    If one accepts the universality of quantum theory including superobservers, we propose a natural extension of Peres's dictum to resolve these extended Wigner's friend paradoxes.
\end{abstract}


\tableofcontents

\section{Introduction}
In 2016, Frauchiger and Renner (FR) \cite{frauchiger2018quantum} proposed a paradox arguing for the inconsistency of the use of quantum theory to describe itself,
by an extension of the Wigner's friend scenario \cite{wigner1995remarks}.
In it, agents -- referred to as `superobservers' -- describe other agents as quantum systems, and the agents reason about each other's knowledge, leading to a contradiction.
At the heart of the paradox lies the quantum measurement problem \cite{baumann2016measurement,vilasini2022general,kastner2020unitary,vilasini2019multi,brukner2018no},
but consensus about what the FR paradox implies for quantum foundations is still lacking.
We aim to sharpen this question by proposing a closely related paradox leading to a stronger no-go theorem, grounded on the strength of the underlying contextuality.

We start from two observations concerning the FR paradox.
First, by slightly modifying the reasoning stage, a stronger paradox is obtained where reasoning is performed only by \emph{classical} agents.\footnote{By a \textit{classical} agent we mean an agent that need not be modelled quantumly by any other agent in the protocol (see also \Cref{sec:Wignerfriend}).}
Notably, the consistency assumption used in the FR paradox may already be violated in classical paradoxes~\cite{jones2024thinking,walleghem2024stunnedsleepingbeautyprince} as we explain more below; this consistency assumption is weakened significantly in this work and is unproblematic in these classical paradoxes.
Second, the essential feature underlying the FR paradox is the logically nonlocal Hardy model \cite{hardy1993nonlocality}, as already recognised in Refs.~\cite{aaronson2018s,drezet2018wigner,vilasini2019multi,montanhano2023contextuality,vilasini2022general,fortin2019wigner,bub2018defensesingleworldinterpretationquantum,Dourdent_2021}. 
Based on these observations, we propose a stronger paradox whose underlying nonlocality argument is the GHZ--Mermin model \cite{mermin1990simple,greenberger1989going}.
Strong nonlocality of this model means that, compared to FR, the new paradox does away with the need for post-selection on measurement outcomes before running the logical argument.

We use this GHZ--FR paradox to provide two no-go theorems of increasing strength, the GHZ--FR \textit{truth} and \textit{agreement} no-go theorems.
In the truth no-go theorem the emphasis lies on the agents’ points of view about an underlying (absolute) truth, whereas in the agreement no-go theorem the underlying truth is replaced by the consistency condition that classical agents, who classically communicate and meet, must agree on outcome values.
The truth no-go theorem connects to the literature on the Absoluteness of Observed Events,
often invoked in Wigner's friend-type paradoxes \cite{bong2020strong,cavalcanti2021implications,haddara2022possibilistic,leegwater2022greenberger,schmid2023review,brukner2018no,ormrod2022no,ormrod2023theories,nurgalieva2018inadequacy,vilasini2019multi,montanhano2023contextuality,zukowski2021physics,walleghem2023extended,walleghem2024connecting,gao2019quantum,guerin2020no,ying2023relating}, and its assumptions are natural and easy to grasp.
If one wishes to resolve this no-go theorem by rejecting the existence of absolute facts, the agreement no-go theorem suggests that the only agents who may assign an outcome to a measurement are those with direct access to it: the experimenter and by extension any agent learning about its result, \ie any agent obtaining a (copy of) the outcome record.
Of course, other resolutions of the paradox circumventing our no-go theorems are possible, for example by refuting the existence of superobservers as in Refs.~\cite{relano2018decoherence,relano2020decoherence,kastner2020unitary,zukowski2021physics,gambini2019single}, or by further restricting the validity of the Born rule as in Refs.~\cite{lazarovici2019quantum,sudbery2017single}.
Protocols similar to the GHZ--FR paradox have been proposed before in Refs.~\cite{zukowski2021physics,leegwater2022greenberger},
but with a different analysis and weaker no-go theorems. 
A detailed comparison with these protocols and a broader discussion about how our contributions fit within the existing literature can be found in \Cref{sec:previous_work} and \Cref{appendix:related_work}.

\paragraph*{Outline} 
This document is organised as follows.
In \Cref{sec:FR_paradox}, we recap the Wigner's friend \cite{wigner1995remarks}  and FR  \cite{frauchiger2018quantum} scenarios and make two essential observations about the FR paradox.
In \Cref{sec:GHZ_final}, we present the GHZ--FR paradox based on the strongly contextual GHZ--Mermin model.
Next, we carefully analyse the required assumptions for this paradox and derive the GHZ--FR truth and agreement no-go theorems in \Cref{sec:GHZ_FR_nogo}, where we also propose a resolution.
In \Cref{sec:comparison_FR}, we directly compare the GHZ--FR paradox to the original FR paradox, briefly review earlier responses to the FR paradox and comment on related work where protocols similar to the GHZ--FR paradox have been proposed.
We conclude in \Cref{sec:conclusion} with discussing implications of the GHZ--FR paradox and directions for further research.
The appendices contain clarifications of concepts and proofs, and gather some background material.

\section{The role of logical contextuality in the FR paradox} \label{sec:FR_paradox}
We review the original Wigner's friend scenario and the Frauchiger--Renner (FR) paradox.
We observe that the FR paradox
(i) can be strengthened by a slight modification in the reasoning stage, and
(ii) is based on the nonlocal Hardy model.
These observations suggest a stronger FR-like paradox that we describe in the next section.

\subsection{Wigner's friend scenario}\label{sec:Wignerfriend}
The Wigner's friend scenario \cite{wigner1995remarks} is the simplest thought experiment exploring the concept of a superobserver.
It involves Wigner and his friend, Bob, who resides inside a sealed lab and measures a quantum system $S_{\!B}$.
Wigner treats Bob's lab itself as a quantum system, which Wigner, as a superobserver of Bob, may measure and control from the outside.
We denote by $L_{\!B}$ the system consisting of Bob's lab, including Bob himself but excluding $S_{\!B}$.
The setup is represented schematically in \Cref{fig:Wignerfriend}.

\begin{figure}[h]
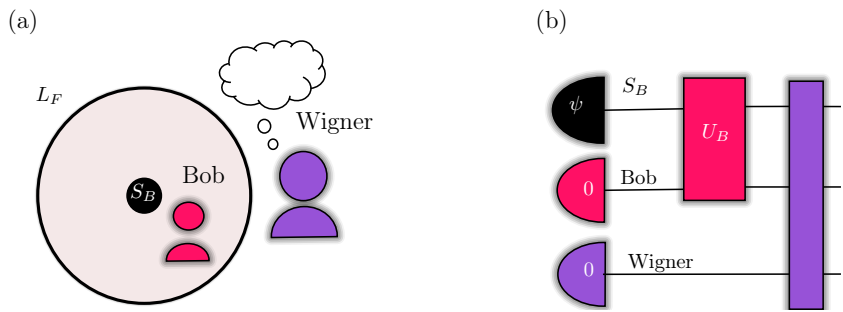
  \centering   \includestandalone[width=0.7\textwidth]{tikzfigures/FINAL_WF} \caption{(a): Schematic sketch of Wigner's friend scenario. Bob performs a measurement on a system $S_{\!B}$ in a sealed lab $L_{\!B}$. Wigner is a superobserver of his friend Bob, modelling Bob's measurement unitarily as $U_B$. The friend sees a definite single-valued outcome, but Wigner describes his friend as being in a superposition of having obtained different outcomes. (b): Circuit representation of a Wigner's friend scenario.} \label{fig:Wignerfriend}
\end{figure}

The quantum system $S_{\!B}$ is initialised in the state $\ket{\psi}_{S_{\!B}} = \alpha \ket{0}_{S_{\!B}} + \beta \ket{1}_{S_{\!B}}$
with $\alpha, \beta \neq 0$.
The friend measures the system $S_B$ in the computational basis $\{ \ket{0}_{S_{\!B}}, \ket{1}_{S_{\!B}} \}$, labelling the outcomes  $0, 1$.
As a \emph{superobserver} for the friend, Wigner models Bob's measurement of the system $S_{\!B}$ as a unitary evolution of the joint quantum system $S_{\!B} L_{\!B}$\footnote{We will often omit tensor product when it is clearly implied, \ie here $S_{\!B} L_{\!B}$ is an abbreviation of $S_{\!B} \otimes L_{\!B}$.}: 
\begin{equation} \label{eq:unitary_friend}
    \ket{\psi}_{S_{\!B}} \ket{r}_{L_{\!B}} \;\;\overset{U_B}{\longmapsto}\;\; 
    \alpha \ket{00}_{S_{\!B} L_{\!B}} + \beta \ket{11}_{S_{\!B} L_{\!B}},
\end{equation}
where $\ket{0}_{L_{\!B}}$ and $\ket{1}_{L_{\!B}}$ correspond to the two possible outcomes as recorded in Bob's memory and elsewhere in his lab, and $\ket{r}_{L_{\!B}}$ denotes the initial (`ready') state of his lab.
From Wigner's perspective, the system $S_{\!B} L_{\!B}$ is in a superposition of
`$S_{\!B}$ is in state $\ket{0}$ and Bob has recorded the measured outcome $0$' and
`$S_{\!B}$ is in state $\ket{1}$ and Bob has recorded the measured outcome $1$'.
Note that the unitary evolution $U_B$ is described (on the relevant degrees of freedom) as a CNOT gate.
The friend sees a definite single-valued outcome, while Wigner would describe his friend as being in a superposition of states corresponding to different outcomes.

Both claims can in principle be tested, as in Deutsch's version of the Wigner's friend argument \cite{deutsch1985quantum}; see also \cite[Section 4]{nurgalieva2020testing}.
Indeed, the fact that the friend has obtained a definite, single-valued outcome can in principle be known to Wigner, namely by the friend passing a note to Wigner stating `I see a definite single-valued outcome' (but crucially without giving away any information about what the outcome actually is).
Wigner's description of the friend as being in a superposition state can also in principle be tested. 
Namely, as a superobserver, Wigner can perform state tomography on the entire friend's lab.\footnote{One way for Wigner to perform such state tomography is the following \cite{nurgalieva2020testing}.
As a superobserver, Wigner can \textit{undo} the friend's measurement, by applying the inverse unitary $U_B^\dagger$ to the friend's lab and the system, thereby resetting them back to their initial state.
Thus, by applying $U_B^\dagger$ and verifying whether the obtained initial state of $S_B L_B$ is indeed $|\psi \rangle_{S_B}|0 \rangle_{L_B}$ Wigner can verify his superposition assignment.}

This thought experiment shows the incompatibility between the existence of superobservers and an absolute, observer-independent notion of collapse \cite{bong2020strong,nurgalieva2018inadequacy,deutsch1985quantum}.
However, even if one takes the wave function collapse not to be absolute, the possibility remains that the underlying fact of the friend's observed outcome be absolute \cite{bong2020strong}.
To see this, suppose Wigner and his friend run the protocol many times. In some rounds Wigner verifies his superposition assignment of \cref{eq:unitary_friend} for the state of his friend Bob's lab, while in other rounds Wigner directly asks his friend Bob for the outcome. 
In that case, the two may disagree on (the timing of) the occurrence of collapse, but they agree on the underlying fact, namely which outcome the friend observed. 
While not asking his friend, Wigner may describe his friend as being in a superposition, but this does not rule out the possibility of there being an absolute fact regarding which outcome the friend obtained, even if unknown to or hidden from Wigner.

This is where extended Wigner's friend scenarios, involving more than one superobserver, come into play.
Since Brukner's scenario \cite{brukner2017quantum}, many different extended Wigner's friend scenarios have been proposed \cite{frauchiger2018quantum,bong2020strong,cavalcanti2021implications,haddara2022possibilistic,leegwater2022greenberger,schmid2023review,brukner2018no,ormrod2022no,ormrod2023theories,nurgalieva2018inadequacy,vilasini2019multi,montanhano2023contextuality,zukowski2021physics,walleghem2023extended,gao2019quantum,guerin2020no,utreras2023allowing,szangolies2020quantum}. 
Among them, the FR paradox \cite{frauchiger2018quantum} has gained much attention.

\subsubsection*{Terminology: (super)observers and (classical) agents} 
For reasons of clarity, we shortly define what we mean by an agent, a measurement and introduce terminology regarding (super)observers and (classical) agents.

\paragraph*{Definition of an agent and measurement} \label{def:agent}
An agent is any system that can obtain information through measurements, and store information. In quantum information theory using the Hilbert space formalism, the ability to store information means that an agent Bob has a personal preferred orthogonal basis in which he can store information. Storing information in these orthogonal `knowledge' states ensure his information states are perfectly distinguishable for him.
Bob needs a preferred basis because one cannot interpret multiple bases as simultaneously encoding information about a quantum system for the same single observer~\cite{brukner2021qubits,allam2023observer}.
By a measurement, we mean an interaction of the agent Bob with a system $S$ that increases the information about $S$ in his memory (\ie in his knowledge states, the orthogonal states in his personal preferred basis).
The agent being able to obtain information thus means that the agent must be able to perform some measurements, and store the outcome in his knowledge states\footnote{Storing an outcome in a preferred basis might lead to production of entropy as in practice it corresponds to a state preparation, see for example \cite{di2021arrow}, and thus might require suitable entropic conditions and an (agential) thermodynamic arrow of time. A preferred basis might also arise from interactions; see for example Ref.~\cite{ormrod2024quantum} for a proposal along these lines.}.

\paragraph*{Terminology and definitions regarding agents} We will use the terms \textit{agent} and \textit{observer} interchangeably. By a \textit{superobserver} we mean an observer who models another agent $A$ quantumly, and can perform arbitrary quantum operations on $A$ and $A$'s whole (closed) lab, including reversing operations previously performed by $A$.
By a \textit{classical} agent we mean an agent that need not be modelled quantumly by any other agent in the protocol,
\ie an agent who is above every other agent's Heisenberg cut in the protocol.
Note that the notion of a \textit{classical agent} is considered relative to a specific protocol $P$:  assuming the existence of superobservers, one can always extend a protocol such that a classical agent in the old protocol is modelled quantumly by another agent in the new protocol and thus is no longer a classical agent in the new protocol.\footnote{Our notion of `classical agent' is closely related to the definition of an `admissible observer' $A$ relative to another observer $O$ and history $h$ of Ref.~\cite{baltag2023logic}, characterised as `[...] (as far as the background observer $O$ can know) [...] none of the information carried by $A$ will be fully erased from the universe at any moment of the given history $h$'. Namely, a classical observer according to our definition is an admissible observer relative to all other observers in the protocol and the history of the protocol.
}

\subsection{The entanglement version of the FR paradox}
\label{sec:FR_overview}
We sketch the entanglement version of the FR paradox as outlined in Ref.~\cite{vilasini2022general}, credited to Lluis Masanes.
For the original prepare-and-measure version of the FR paradox we refer to Refs.~\cite{frauchiger2018quantum,vilasini2022general}.

The scenario is depicted in \Cref{fig:FR_circuit}(a).
It involves two superobservers, Ursula  and Wigner, each with their own `friend', Alice and Bob respectively.
Each of the friends reside in a sealed lab with access to quantum systems $S_{\!A}$ and $S_{\!B}$ respectively, which are initialised in a shared entangled state.
As before, $L_{\!A}$ denotes Alice's lab including Alice herself but excluding the system $S_{\!A}$, and similarly $L_{\!B}$ for Bob.
All the agents -- Alice, Bob, Ursula, and Wigner -- perform measurements, obtaining outcomes $a,b,u,w \in \{0,1\}$ respectively.

\paragraph*{Measurement protocol} 
The measurement protocol, sketched in circuit form in \Cref{fig:FR_circuit}(b), proceeds as follows;
\begin{enumerate}[label=\arabic*.,leftmargin=*]
    \item \textbf{Initialisation.} The system $S_{\!A} S_{\!B}$, shared between Alice and Bob, starts in the specific entangled state $|\psi_0\rangle_{S_{\!A} S_{\!B}} = \frac{1}{\sqrt{3}} \left(\ket{00}+\ket{10}+\ket{11}\right)_{S_{\!A} S_{\!B}}$.
     The initial states of Alice's and Bob's labs are labelled by $\ket{r}_{L_{\!A}}$ and $\ket{r}_{L_{\!B}}$, respectively.
     The initial state of $S_{\!A} L_{\!A} S_{\!B} L_{\!B}$ is given by\footnote{\label{fn:tensorproduct} For notational ease, we will often omit the tensor product symbol: $\ket{00}_{AB}=\ket{0}_A \ket{0}_B = \ket{0}_A \otimes \ket{0}_B$.}
    \begin{align}
    \begin{split}
    \ket{\psi^{t=1}}_{S_{\!A} L_{\!A} S_{\!B} L_{\!B}}
    \;=\;&
    \frac{1}{\sqrt{3}}\left(\ket{0 r 0 r}+\ket{1r0r}+\ket{1r1r}\right)_{S_{\!A} L_{\!A} S_{\!B} L_{\!B}}.
    \end{split}
    \end{align}

\item \textbf{Alice and Bob measure.} Alice measures her system $S_{\!A}$ in the $\{ \ket{0}, \ket{1} \}$ basis, 
    and Bob similarly measures his system $S_{\!B}$.
    The unitary description of these measurements, as in \cref{eq:unitary_friend}, leads to the overall state
    \begin{align}\label{eq:FR_psi_t2}
    \begin{split}
    \ket{\psi^{t=2}}_{S_{\!A} L_{\!A} S_{\!B} L_{\!B}}
    \;=\;& \left(U_{A} \otimes U_{B} \right)\ket{\psi^{t=1}}_{S_{\!A} L_{\!A} S_{\!B} L_{\!B}}
     \\ \;=\;& \frac{1}{\sqrt{3}}(\ket{0000}+\ket{1100}+\ket{1111})_{S_{\!A} L_{\!A} S_{\!B} L_{\!B}}.
    \end{split}
    \end{align}
    where $U_{A}$ and $U_{B}$ are the unitaries describing Alice's and Bob's measurements, respectively.
    \item \textbf{Ursula and Wigner supermeasure.} Finally, the superobservers Ursula and Wigner respectively measure the systems $S_{\!A} L_{\!A}$ and $S_{\!B} L_{\!B}$ in a 
    basis with vectors\footnote{Note that a basis requires two more orthonormal vectors, but their corresponding outcomes have probability zero, so we do not need to specify the remaining vectors. \label{fn:fullbasis}}
    \begin{equation} \ket{\textok}=\frac{1}{\sqrt{2}}(\ket{00}-\ket{11}), \qquad \ket{\textfail}=\frac{1}{\sqrt{2}}(\ket{00} + \ket{11}) .\end{equation}
\end{enumerate}
The protocol is post-selected on both Ursula and Wigner obtaining the outcome $\textok$, denoted $u = \textok$ and $w = \textok$.
Observe that this happens with nonzero probability:
\begin{align}
&\abs*{\bra{\textok}_{S_{\!A} L_{\!A}}\otimes\bra{\textok}_{S_{\!B} L_{\!B}}\ket{\psi^{t=2}}_{{S_{\!A} L_{\!A} S_{\!B} L_{\!B}}}}^2
\\
\;=\; & \abs*{\frac{1}{2\sqrt{3}} (\bra{0000}-\bra{0011}-\bra{1100}+\bra{1111} ) (\ket{0000}+\ket{1100}+\ket{1111})}^2
\\ \;=\; & \abs*{\frac{1}{2\sqrt{3}}}^2 = \frac{1}{12}  \;>\; 0. 
\end{align}

\paragraph*{Reasoning steps}
After a successful post-selected run of the protocol, Ursula reasons (i) about Bob's outcome $b$, (ii) about Bob's conclusion for Alice's outcome $a$, and (iii) about Bob's reasoning about Alice's reasoning about Wigner's outcome $w$.
Assuming that each agent uses the Born rule to make predictions, Ursula, using $u = \textok$, reasons that Bob must have obtained $b=1$, and thus that he would have predicted that $a=1$, so that Alice would have concluded that $w = \textfail$.
In summary,
\begin{equation} u = \textok \;\Rightarrow\; b = 1 \;\Rightarrow\; a = 1 \;\Rightarrow\; w = \textfail, \end{equation}
contradicting the initial post-selection on $u = \textok, w = \textok$.
We give additional details in the next section where we show how the argument may in fact be strengthened.

\begin{figure}
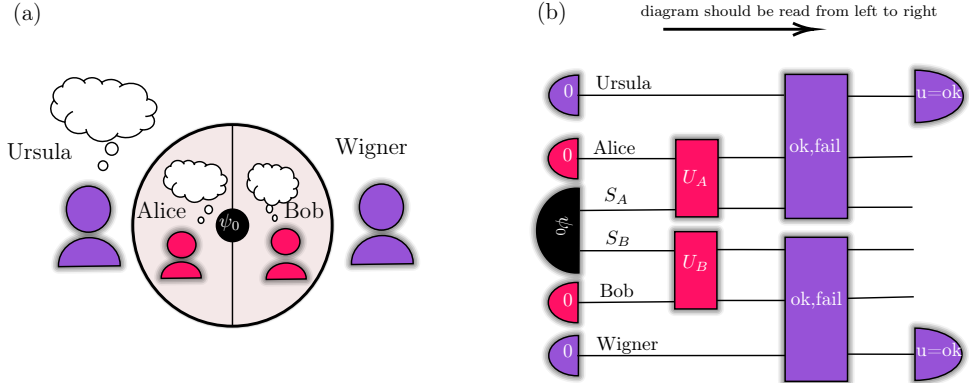
 \centering \includestandalone[width=0.8\textwidth]{tikzfigures/FINAL_FR2}   \caption{(a): Sketch of the FR paradox, where Ursula reasons about Bob who reasons about Alice who reasons about Wigner. (b): Circuit of the entanglement version of the FR paradox, as outlined in \cite{vilasini2022general}, with the initial state of $S_{\!A}\otimes S_{\!B}$ given by $|\psi_0\rangle_{S_{\!A} S_{\!B}} = \frac{1}{\sqrt{3}} \left(\ket{00}+\ket{10}+\ket{11}\right)_{S_{\!A} S_{\!B}}$. 
} \label{fig:FR_circuit} 
\end{figure}

\subsection{A slightly stronger modified FR paradox}
\label{sec:stricter_FR}
We provide a stronger variant of the FR paradox in which only \emph{classical}\footnote{For a definition of a \emph{classical} agent, see \Cref{sec:Wignerfriend}.}
 agents reason.
This technique will be used in the GHZ--FR paradox as well.

The measurement protocol is identical to that of the FR paradox.
But an additional agent, Zeno,\footnote{The additional agent is needed here to reason about the joint outcomes of the friends Alice and Bob, but will not be necessary in the case of the GHZ--FR paradox of \Cref{sec:GHZ_final}.} who does not perform any measurement, makes a prediction for Alice's and Bob's outcomes.
Ursula and Wigner, after obtaining their outcomes, reason about the outcomes obtained by Alice and Bob, respectively.
Ursula, and Wigner hand their predictions to Zeno, who derives a contradiction.\footnote{We note that Ref.~\cite{healey2018quantum} gives a presentation of the reasoning steps of the FR paradox that comes close to the one here.}

Post-selecting on $u=\textok,w=\textok$, which has a non-zero probability of occurring, we have the following predictions from applications of the Born rule: 
\begin{itemize}
    \item Ursula predicts that $b=1$ by excluding $b=0$ since 
\[\bra{\textok}_{S_{\!A} L_{\!A}} \bra{00}_{S_{\!B} L_{\!B}} \ket{\psi^{t=2}}_{S_{\!A} L_{\!A} S_{\!B} L_{\!B}}=0\,.\]

    \item Wigner predicts that $a=0$ by excluding $a=1$ since 
    \[\bra{11}_{S_{\!A} L_{\!A}}\bra{\textok}_{S_{\!B} L_{\!B}} \ket{ \psi^{t=2}}_{S_{\!A} L_{\!A} S_{\!B} L_{\!B}}=0\,\]
    \item Zeno predicts that $a=0, b=1$ can never occur since 
    \[\bra{00}_{S_{\!A} L_{\!A}}\bra{11}_{S_{\!B} L_{\!B}} \ket{\psi^{t=2}}_{S_{\!A} L_{\!A} S_{\!B} L_{\!B}}=0\,. \]
\end{itemize}
Ursula and Wigner classically communicate their predictions to Zeno.
Upon combining the predictions made by Ursula and Wigner with his own, Zeno finds a contradiction.
Thus, all reasoning and predicting is done by classical agents. 
Furthermore, each application of the Born rule refers to measurement outcomes that are obtained concurrently: its predictions could in principle be tested by interrupting the protocol before one or both supermeasurements.

\subsection{Logical non-locality powers the FR paradox}\label{sec:FR_non-locality}
In \Cref{sec:stricter_FR} we have seen how actually a contradiction can be obtained by having only \emph{classical} agents reason in the FR protocol.
The second crucial observation about the FR paradox is to recognise that it is underpinned
by the Hardy model \cite{hardy1993nonlocality}, a proof of ``nonlocality without inequalities''. 
This approach has a logical flavour in that nonlocality is witnessed in terms of \emph{possibilities} alone rather than requiring full knowledge of the \emph{probabilities}.
In fact, Frauchiger and Renner~\cite{frauchiger2018quantum} explicitly stated that their protocol is motivated by Hardy's paradox, a connection that has also been emphasised in subsequent works such as Refs.~\cite{frauchiger2018quantum,aaronson2018s,drezet2018wigner,vilasini2019multi,montanhano2023contextuality,vilasini2022general,fortin2019wigner}.

The Hardy model is a specific quantum realisation of the simple $(2,2,2)$ Bell scenario. This nonlocality scenario involves two parties, each able to choose between two measurement settings, with two outcomes $0$ and $1$.
The Hardy model, on which the FR paradox is based, arises from a specific choice of initial shared state and local measurements in the $(2,2,2)$ Bell scenario.
The two parties share a system in the entangled state  $\sfrac{1}{\sqrt{3}} (\ket{00}+\ket{10}+\ket{11})$,
which each can choose to measure locally in either the $\{\ket{0},\ket{1}\}$ or the $\{\ket{+},\ket{-}\}$ basis, known respectively as the $Z$- and the $X$-basis.\footnotemark

\Cref{fig:FR_Rui} depicts a mapping between the $(2,2,2)$ Bell scenario and the extended Wigner's Friend scenario of the FR paradox.

\begin{figure}
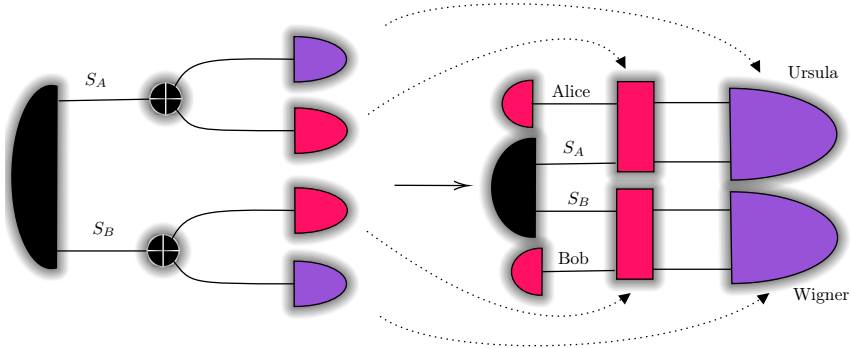
 \centering 
\includestandalone[width=0.7\textwidth]{tikzfigures/tikz_FR_to_Bell}  \caption{Mapping a Bell non-locality scenario (left) to an FR-like extended Wigner's friend scenario (right). In the Bell scenario each party can choose (pictured by a crossed circle) one out of two (incompatible) measurements, for example in the $Z$ basis $\{\ket{0},\ket{1}\}$ or in the $X$ basis $\{\ket{+},\ket{-}\}$.
In the corresponding extended Wigner's friend scenario, for each Bell party, one of the measurements, for example the computational basis $\ket{0},\ket{1}$ measurement, is performed by an observer and modelled unitarily as $U$, whereas the other measurement is performed by the corresponding superobserver in an entangled basis $U |\pm \rangle_{S_{\!F}} |0\rangle_{L_{\!F}}$.} \label{fig:FR_Rui} 
\end{figure}

The key idea is better explained by focusing on a \emph{single} party.
That is an experimenter who, in each run of the experiment, chooses between two (incompatible) dichotomic measurements to perform on their part of a shared system.
Such a situation can be translated into a Wigner's friend scenario:
the friend first performs one of the measurements, and then the superobserver Wigner undoes this measurement and performs the other.
The choice of measurement setting is thus translated to a choice between adopting the point of view of Wigner or that of the friend.
Applying the same idea to each party, the translation extends to the $(n,2,2)$ Bell scenario for any number $n$ of parties.

Bell nonlocality scenarios are a particular case of contextuality scenarios, where maximal measurement contexts correspond to a choice of one measurement for each party.
In the $(2,2,2)$ Bell scenario there are four dichotomic measurements, two for each party. 
We choose to label Alice's measurements as $A$ and $U$ and Bob's as $B$ and $W$, foreshadowing the link to the FR scenario. The maximal measurement contexts are then
\[\{A,B\}, \quad \{U,B\}, \quad \{A,W\}, \quad \{U,W\} .\]
Using the Born rule we calculate the (two-party) joint outcome probabilities, yielding a probability distribution over joint outcomes for each context. 
In fact, in this case, the \emph{possibilistic} version of the Born rule suffices, which only distinguishes between \emph{possible} (probability greater than zero) and \emph{impossible} (probability zero) events.
\Cref{table:FR_table} indicates such possibilistic empirical data for the specific choice of state and measurements in the Hardy model:
each row is labelled by a measurement context and indicates the support of the probability distribution on the joint outcomes for that context; see \eg~\cite{abramsky2011sheaf,abramsky2015contextuality}. 

\footnotetext{We regard $\ket{0}$, $\ket{+}$ as corresponding to the outcome label $0$ and $\ket{1}$, $\ket{-}$ to the label $1$. Note that these are, respectively, the $+1$ eigenvectors and the $-1$ eigenvectors of the $Z$ and $X$ matrices.
Often, when regarding the $Z$ or $X$ matrices as measurement operators, one considers outcomes valued $\pm1$. Here, we use the bijection $+1 \mapsto 0$, $-1 \mapsto 1$.
Particularly relevant for the GHZ--Mermin argument, this bijection actually yields a group isomorphism between $\{+1,-1\}$ under multiplication and $\{0,1\}$ under addition modulo $2$, an instance of the isomorphism between the multiplicative group of complex $n$-th roots of unity  and the additive group of integers modulo $n$.\label{fn:01pm1}}

\newcommand{\tikzxmark}{%
\tikz[scale=0.18] {
    \draw[line width=0.7,line cap=round] (0,0) to [bend left=6] (1,1);
    \draw[line width=0.7,line cap=round] (0.2,0.95) to [bend right=3] (0.8,0.05);
}}
\newcommand{\tikzcmark}{%
\tikz[scale=0.18] {
    \draw[line width=0.7,line cap=round] (0.25,0) to [bend left=10] (1,1);
    \draw[line width=0.8,line cap=round] (0,0.35) to [bend right=1] (0.23,0);
}}

\begin{table}
\centering 
{%
\begin{tabular}{l|cccc|}
\cline{2-5}

                                & \footnotesize{$0,0$}            & \footnotesize{$0,1$}             & \footnotesize{$1,0$}             & \footnotesize{$1,1$}            \\ \hline
\multicolumn{1}{|l|}{$A,B$} & $1$ & 0 & $1$            & $1$  \\ 
\multicolumn{1}{|l|}{$A,W$} & $1$ & $1$             & 0  & $1$  \\ 
\multicolumn{1}{|l|}{$U,B$} & 0 & $1$ & $1$  & $1$               \\ 
\multicolumn{1}{|l|}{$U,W$} & $1$ & $1$ & $1$ & $1$ \\ \hline
\end{tabular}%
}%
\caption{Support of joint outcome probability distributions for each context in the Hardy model or the FR paradox.
The rows are labelled by the four measurement contexts $\{A,B\},\{A,W\},\{U,B\},\{U,W\}$ and columns by the joint outcomes $(0,0),(0,1),(1,0),(1,1)$,
the combination of which gives an event: a specific joint outcome for a specific context of measurements.
The entries in each cell of this possibility table have two possible values: $1$ denotes an event being \textit{possible} while $0$ denotes an event being \textit{impossible}, \ie having zero probability.
Note that in the FR case the outcomes $u$ and $w$ are labelled by $\textok,\textfail$ instead of $0,1$.
This  table is logically contextual:
there is no consistent global assignment of values to $A,B,U,W$ that extends the possible event $\{U \mapsto 0, W \mapsto 0\}$ (in the bottom left);
starting from this partial assignment, row 3 forces $B \mapsto 1$, which used in row 1 gives $A \mapsto 1$, which used in turn in row 2 yields
$W \mapsto 1$, contradicting the initial assignment.} 
~

\label{table:FR_table}
\end{table}

Contextuality, of which non-locality is a special instance, occurs when there is no global probability distribution (on assignments of outcomes to all the measurements in the scenario) that marginalises to the empirical probability distributions within each measurement context.
The Hardy model \cite{hardy1993nonlocality} witnesses \emph{logical} non-locality,
in that non-locality can be ascertained from the possibilistic information collected in \Cref{table:FR_table} alone, \ie at the level of support of outcome probability distributions.
The argument corresponds precisely to the logical reasoning chain in the FR paradox.
Specifically, the joint outcome assignment $\{u \mapsto 0, w \mapsto 0\}$, which occurs with positive probability,
cannot be extended to a global assignment for all $a,b,u,w$ which does not restrict to an impossible event in some other context (see caption of \Cref{table:FR_table} for details).
We note that the troublesome global assignments in the contextuality setting correspond to statements about (absolute) outcomes in FR-like scenarios.

The mapping from a Bell scenario to an extended Wigner's friend scenario as in \Cref{fig:FR_Rui}
provides a recipe for producing extended Wigner's friend paradoxes based on nonlocal models.
We will exploit this recipe in the next section to produce an FR-like paradox based on the GHZ--Mermin model \cite{greenberger1989going,mermin1990simple}.

\section{The GHZ--FR paradox} \label{sec:GHZ_final}
In this section, we first recap the GHZ--Mermin model, the simplest quantum example of strong nonlocality.
We then use it to construct  the GHZ--FR paradox, consisting of a measurement protocol and reasoning steps.
In contrast with the FR paradox, the GHZ--FR paradox requires no post-selection and only classical agents reason and make predictions.

\subsection{Strong contextuality: the GHZ--Mermin argument} \label{sec:GHZ_recap}
\paragraph*{Strong contextuality}
The sheaf-theoretic approach to non-locality and contextuality \cite{abramsky2011sheaf,abramsky2015contextuality,barbosa2015contextuality} introduced a hierarchy of contextuality of increasing strength: probabilistic, logical (or possibilistic), and strong contextuality.
In quantum theory, each of these levels of the hierarchy is exemplified by a well-known model:
\[
    \text{Bell--CHSH (probabilistic)} \;\prec\; \text{Hardy (logical)} \;\prec\; \text{GHZ--Mermin (strong)} .
\]

The contextual model underlying the FR paradox is the logically contextual Hardy model: contextuality is witnessed at the level of the \emph{support} of probability distributions. Namely, there exists a \emph{possible} assignment of outcomes to a context which cannot be extended to a global assignment (to all measurements) consistent with the model, \ie whose restriction to each context is deemed possible.

Nevertheless, in the Hardy model, there exist a local assignment that can be extended globally; \eg the global assignment $\{a \mapsto 1, b \mapsto 1, u \mapsto 1, w \mapsto 1\}$ is consistent with the observed possibilities in \Cref{table:FR_table}.
\emph{Strong contextuality} holds when not a single such global assignment exists.\footnotemark

\footnotetext{This property coincides with \emph{maximal contextuality} in the sense that the fraction of model that can be explained by noncontextual hidden variables is zero, or that the algebraic (no-signalling) maximum of a noncontextuality inequality violation is attained \cite{abramsky2017contextual}.}

\paragraph*{The GHZ--Mermin model}
A paradigmatic example of a quantum-realisable strong contextuality is given by the GHZ--Mermin model \cite{greenberger1989going,mermin1990quantum}, a specific realisation of the $(3,2,2)$ Bell scenario. Following the presentation in Ref.~\cite{abramsky2015contextuality}, we consider three parties, $A$, $B$, and $C$, that share an entangled three-qubit system prepared in the GHZ state,
\begin{equation}
\ket{ \psi^{\textsf{GHZ}} }_{ABC}  = \frac{1}{\sqrt{2}} \left(\ket{000}+\ket{111}\right)_{ABC}.
\end{equation} 
Each party performs either an $X$ or a $Y$ measurement on their qubit,
\ie a measurement in either the $\{\ket{+} , \ket{-}\}$ or the $\{\ket{+i}, \ket{-i}\}$ basis,
where $\ket{\pm } = (\ket{0} \pm \ket{1})/\sqrt{2}$ and $\ket{\pm i} = (\ket{0} \pm i\ket{1})/\sqrt{2}$. This leads to measurements $X_A,Y_A,X_B,Y_B,X_C,Y_C$ with outcomes $x_A, y_A, x_B, y_B, x_C, y_C \in \{0,1\}$.
The outcomes corresponding to $\ket{+}$ and $\ket{+i}$ are labelled $0$, those corresponding to  $\ket{-}$ and $\ket{-i}$ are labelled $1$.

A measurement context in this scenario consists of a choice of measurement for each party. We abbreviate the context $\{X_A, X_B, X_C\}$ as $XXX$ and analogously for the others.
Using the Born rule to distinguish possible from impossible events for measurement contexts $XXX$, $XYY$, $YXY$ and $YYX$, one obtains \Cref{table:GHZ}.
Each row is labelled by a measurement context and indicates the support of a probability distribution on the respective joint outcomes; see \eg~\cite{abramsky2011sheaf,abramsky2015contextuality}.
The outcomes satisfy the parity equations
\begin{equation} \label{eq:GHZ_Mermin}
x_A \oplus x_B \oplus x_C=0, \quad
x_A \oplus y_B \oplus y_C=1, \quad
y_A \oplus x_B \oplus y_C=1, \quad
y_A \oplus y_B \oplus x_C=1, 
\end{equation}
where $\oplus$ denotes addition modulo $2$.\footnote{This argument is also often presented with multiplication of $\pm1$-labelled outcomes; see \cref{fn:01pm1}.}
Assigning definite outcome values to all six measurements simultaneously, summing these four equations yields $0=1$, a contradiction. 
Therefore, no global assignment of outcomes to all six measurements is compatible with the possible observations for all measurement contexts -- the GHZ--Mermin model is strongly contextual.

\begin{table}
\centering
{%
\begin{tabular}{l|cccccccc|}
\cline{2-9}
                                & \footnotesize{$0,0,0$}            & \footnotesize{$0,0,1$}  & \footnotesize{$0,1,0$} & \footnotesize{$0,1,1$} & \footnotesize{$1,0,0$} & \footnotesize{$1,0,1$} & \footnotesize{$1,1,0$} & \footnotesize{$1,1,1$}            \\ \hline
\multicolumn{1}{|l|}{$XXX$} & $1$ & $0$ & $0$ & $1$ & $0$ & $1$ & $1$ & $0$  \\ 
\multicolumn{1}{|l|}{$XYY$} & $0$ & $1$ & $1$ & $0$ & $1$ & $0$ & $0$ & $1$  \\ 
\multicolumn{1}{|l|}{$YXY$} & $0$ & $1$ & $1$ & $0$ & $1$ & $0$ & $0$ & $1$  \\ 
\multicolumn{1}{|l|}{$YYX$} & $0$ & $1$ & $1$ & $0$ & $1$ & $0$ & $0$ & $1$ \\ \hline
\end{tabular}%
}%
\caption{Table  with outcome (im)possibilities for measurements by three parties who share a GHZ state. The rows correspond to the support of probability distributions over the outcomes (columns 1--4) that arise when the measurements specified in the first column are performed ($XXX,XYY,YXY,YYX$). 
}
\label{table:GHZ}
\end{table}

\subsection{Description of the GHZ--FR paradox} \label{sec:GHZ_protocol}

Mapping Bell nonlocality scenarios to extended Wigner's friend scenarios, as described in \Cref{sec:FR_non-locality}, we build an FR-like paradox based on the GHZ--Mermin model.
The strong contextuality of this nonlocal model does away with the need for post-selection in the GHZ--FR paradox.

The scenario is depicted in \Cref{fig:GHZ_final}(a). It involves three observer-superobserver pairs -- Alice and Ursula, Bob and Valentina, and Charlie and Wigner -- corresponding to the three parties in the GHZ--Mermin model.
The friends each reside in a sealed lab with access to quantum systems $S_{\!A}$, $S_{\!B}$, and $S_{\!C}$ respectively, which are initialised in a shared entangled state.
As before, $L_{\!A}$ denotes Alice's lab including Alice herself but excluding the system $S_{\!A}$, and similarly $L_{\!B}$ and $L_{\!C}$ for Bob and Charlie.

Alice, Bob, Charlie and their respective superobservers Ursula, Valentina, and Wigner perform measurements with outcomes $a,b,c,u,v,w \in \{0,1\}$. 
A circuit representation of the measurement protocol is shown in \Cref{fig:GHZ_final}(b).
First, the observers $A, B, C$ perform $Y$-measurements on their part of a GHZ state. 
Then, Ursula performs an $X$-like measurement on Alice and her qubit,
and writes down her predictions for the outcomes of $B, C$.
One intuition for why it might be natural for Ursula to formulate predictions about $b,c$ is in case she does not know about Valentina's and Wigner's actions (or existence even); in that case she believes to still be able to gather Bob's outcome record, for example by asking him for his outcome. 
Valentina and Wigner act analogously.
Once the protocol has been carried out, Ursula, Valentina and Wigner meet up and share their predictions, deriving a contradiction from them, as we will explain below.

\begin{figure}
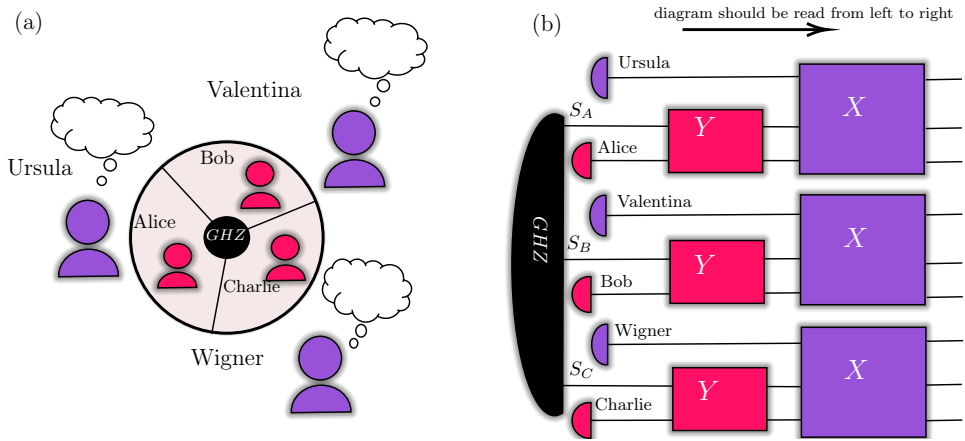
 \centering 
\includestandalone[width=0.8\textwidth]{tikzfigures/FINAL_GHZ_WF}  \caption{(a): Pictorial representation of the scenario GHZ--FR protocol, with superobservers $U,V,W$ and respective observers $A,B,C$, each having access to part of a GHZ state. (b): The GHZ--FR protocol in circuit form. Alice, Bob and Charlie each have access to one part of a three-qubit GHZ state and measure their qubit in the $Y$-basis. Ursula, Valentina and Wigner (red), superobservers for Alice, Bob and Charlie, respectively, each perform an $X$-like measurement on Alice, Bob and Charlie and their qubits.} \label{fig:GHZ_final}
\end{figure}

We first describe the measurement protocol in more detail, after which we present the reasoning steps through which they derive a contradiction.

\paragraph*{Measurement protocol}
The measurement protocol, sketched in circuit form in \Cref{fig:GHZ_final}, proceeds as follows:
\begin{enumerate}[label=\arabic*.,leftmargin=*]
\item \textbf{Initialisation.} The system $S_{\!A} S_{\!B}  S_{\!C}$, shared between Alice, Bob, and Charlie, starts in the GHZ state $|\psi^{\textsf{GHZ}}\rangle_{S_{\!A} S_{\!B}  S_{\!C}}=\frac{1}{\sqrt{2}} \left(\ket{000}+\ket{111}\right)_{S_{\!A} S_{\!B} S_{\!C}}$. The initial states of the three labs are labelled by $\ket{r}_{L_{\!A}}$, $\ket{r}_{L_{\!B}}$, $\ket{r}_{L_{\!C}}$ (with $r$ from `ready').
The initial state of $S_{\!A} L_{\!A} S_{\!B} L_{\!B} S_{\!C} L_{\!C}$ is thus given by
    \begin{equation}
    \begin{split}
    \ket{\psi^{t=1}}_{S_{\!A} L_{\!A} S_{\!B} L_{\!B} S_{\!C} L_{\!C}} = 
    \frac{1}{\sqrt{2}}\left(\ket{0r0r0r}+\ket{1r1r1r}\right)_{S_{\!A} L_{\!A} S_{\!B} L_{\!B} S_{\!C} L_{\!C}}.
    \end{split}
    \end{equation}

\item \textbf{Alice, Bob, and Charlie measure.} Alice, Bob and Charlie perform a measurement in the $Y$-basis $\ket{\pm i}$ on their respective qubit $S_{\!A},S_{\!B},S_{\!C}$, with outcomes $0,1$. Superobserver Ursula describes Alice's measurement by the unitary evolution $U_A$ that,  similar to \Cref{eq:unitary_friend} must satisfy: \begin{equation} \label{eq:GHZ_best_unitary}
    \begin{split}
        \ket{+i }_{S_{\!A}} \ket{r}_{L_{\!A}} \overset{U_A}{\longmapsto} \ket{+i}_{S_{\!A}} \ket{0}_{L_{\!A}}, \qquad
        \ket{-i }_{S_{\!A}} \ket{r}_{L_{\!A}} \overset{U_A}{\longmapsto} \ket{-i }_{S_{\!A}} \ket{1}_{L_{\!A}}
    \end{split}
\end{equation} and similarly for Bob and Charlie where the states $\ket{0}_{L_{\!A}},\ket{1}_{L_{\!A}}$ denote the states of Alice's lab in which she obtains outcomes 0 and 1, respectively. This leads to the overall state
\begin{equation}
    \ket{\psi^{t=2}}_{S_{\!A} L_{\!A} S_{\!B} L_{\!B} S_{\!C} L_{\!C}}
    \;=\; \left(U_{A} \otimes U_{B} \otimes U_{C} \right)\ket{\psi^{t=1}}_{S_{\!A} L_{\!A} S_{\!B} L_{\!B}S_{\!C} L_{\!C}} .
\end{equation}

\item \textbf{Ursula, Valentina, and Wigner supermeasure.}
Finally, the superobservers Ursula, Valentina, and Wigner respectively measure the systems 
$S_{\!A} L_{\!A}$, $S_{\!B} L_{\!B}$, and $S_{\!C} L_{\!C}$ in a basis with vectors\footnote{The two other basis vectors need not be specified since $\ket{+i, 1}_{S_{\!A} L_{\!A}}$  and $\ket{-i, 0}_{S_{\!A} L_{\!A}}$ do not occur.}
    \begin{equation}
     \ket{\textyes}_{S_{\!A} L_{\!A}} =  U_A \ket{+}_{S_{\!A}} \ket{0}_{L_{\!A}}, \qquad \ket{\textno}_{S_{\!A} L_{\!A}} = U_A \ket{-}_{S_{\!A}} \ket{0}_{L_{\!A}} .   
    \end{equation}
 We label the outcomes of $\ket{\textyes}, \ket{\textno}$ by $0,1$.

\end{enumerate}

The probabilities for Ursula's measurement in the $| \textyes \rangle, | \textno\rangle$-basis are equal to those which one would obtain by undoing Alice's measurement on $S_{\!A}$, followed by an $X$-mea\-sure\-ment on the qubit $S_{\!A}$, and similarly for Valentina's and Wigner's measurements. The probabilities for $| \textyes \rangle, | \textno \rangle$ measurements by superobservers thus correspond to probabilities for $X$-measurements on the original GHZ state.
For example, the unitary description of Alice's measurement followed by Ursula's measurement effect yields the probabilities
\begin{equation}
     \abs*{\bra{\pm}_{S_{\!A}} \bra{0}_{L_{\!A}} U_A^\dagger U_A \ket{\ldots}_{S_{\!A}}\ket{0}_{L_{\!A}}}^2 = \abs*{\bra{\pm}_{S_{\!A}}  \bra{0}_{L_{\!A}} \ket{\ldots}_{S_{\!A}}\ket{0}_{L_{\!A}} }^2 = \abs*{\bra{\pm}_{S_{\!A}} \ket{\ldots}_{S_{\!A}}}^2 .
\end{equation}

Possible lightcone structures of this scenario are sketched in \Cref{fig:FR_GHZ_lightcone}.
Our results are independent of the exact causal structure, but we consider the natural causal structure arising from relativistic spacetime.
More generally, we could embed the GHZ--FR scenario into a process-theoretic framework like Categorical Quantum Mechanics \cite{abramsky2009categorical,coecke_kissinger_2017} or the process operator formalism \cite{oreshkov2012quantum}.

\begin{figure}[H]
\centering 
\includestandalone[width=0.8\textwidth]{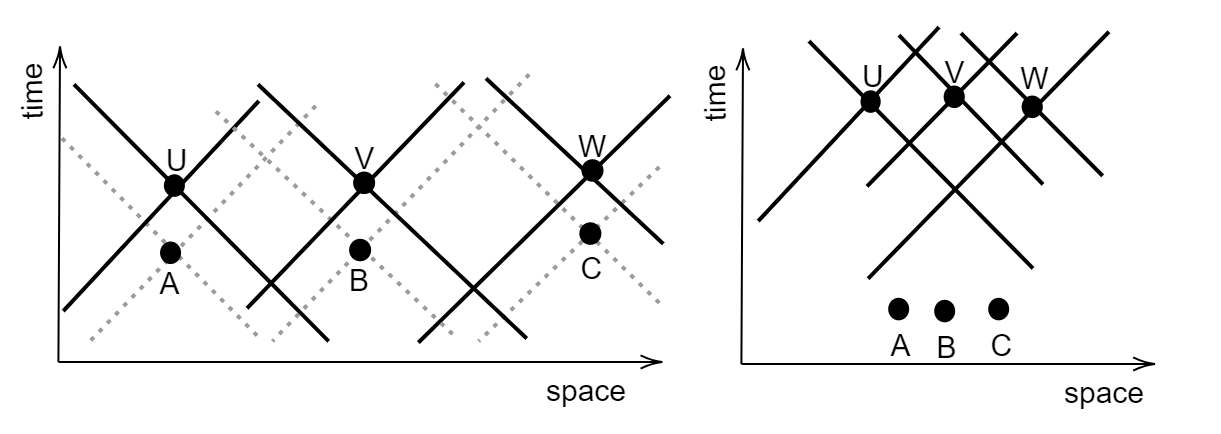}
\caption{Possible lightcone structures of the agents in the GHZ--FR scenario, with Ursula, Valentina, and Wigner being superobservers for, respectively, Alice, Bob, and Charlie, who share a GHZ state. The GHZ--FR paradox occurs in both lightcone structures. A mixture of these lightcone structures is also possible. } \label{fig:FR_GHZ_lightcone} 
\end{figure}

\paragraph*{Reasoning steps}
Using the possibilistic Born rule, Ursula reasons about her outcome together with Bob's and Charlie's outcomes, and concludes: ``\textit{Having obtained $u$, I know that $b\oplus c=1 \oplus u$}", since jointly considering the outcomes of Ursula, Bob and Charlie corresponds to the context $XYY$ in the GHZ--Mermin model.
For example, if Ursula obtained $u=0$ then she finds that $b=c=0$ is excluded as: \begin{equation} \begin{split}
    p(u=0,&b=0,c=0) = \operatorname{Tr}\ket{\Psi} \bra{ \Psi } = 0, \\ & \text{where } \ket{\Psi } = \bra{-}_{S_{\!A}} \bra{0}_{L_{\!A}} U_A^\dagger U_A \bra{ +i }_{S_{\!B}} \bra{  +i }_{S_{\!C}}  \ket{\psi^{GHZ}}_{S_{\!A} S_{\!B} S_{\!C}} \ket{000}_{L_{\!A}L_{\!B}L_{\!C}}.   
\end{split}
\end{equation}
Valentina and Wigner make similar reasoning statements, leading to \begin{equation} \label{eq:abc_predictions_GHZ_FR1}
        \begin{split}
            b \oplus c = 1 \oplus u, \quad
            a \oplus c = 1 \oplus v, \quad
            a \oplus b = 1 \oplus w.
        \end{split}
\end{equation} 
Applying the possibilistic Born rule to the  superobservers' measurements with classically recorded outcomes, one obtains \begin{equation} \label{eq:uvw_GHZ}
    u \oplus v \oplus w = 0.
\end{equation}
As Ursula, Valentina and Wigner come together to share their predictions, they find \cref{eq:abc_predictions_GHZ_FR1,eq:uvw_GHZ}. 
These equations correspond to \cref{eq:GHZ_Mermin} from the GHZ--Mermin scenario, with $x_A,x_B,x_C,y_A,y_B,y_C$ substituted for $u,v,w,a,b,c$. Thus they find a contradiction.\footnote{Note that post-selected on $u,v,w$, the GHZ--FR argument gives a Specker's triangle \cite{liang2010specker}. For instance, 
    for $u=v=w=0$ we find \begin{equation}
        b \oplus c = 1, \quad a \oplus c = 1, \quad a \oplus b = 1.
    \end{equation}
}

\section{Two GHZ--FR no-go theorems} \label{sec:GHZ_FR_nogo}
Having described the GHZ--FR paradox, we spell out the required assumptions in detail that lead to a contradiction.
We construct two no-go theorems of increasing strength, which we name the GHZ--FR \textit{truth} and \textit{agreement} no-go theorems. 
The truth no-go theorem is easier to grasp as its assumptions are more natural.
Moreover, it connects to a wider body of literature as it involves the assumption of Absoluteness of Observed Events, stating that performed measurements have absolute, single-valued outcomes, which is often invoked in extended Wigner's friend paradoxes \cite{bong2020strong,cavalcanti2021implications,haddara2022possibilistic,leegwater2022greenberger,schmid2023review,brukner2018no,ormrod2022no,ormrod2023theories,nurgalieva2018inadequacy,vilasini2019multi,montanhano2023contextuality,zukowski2021physics,walleghem2023extended,walleghem2024connecting,gao2019quantum,guerin2020no}.
The agreement no-go theorem replaces Absoluteness of Observed Events by a weaker notion of agreement among classical agents. 
The fact that the paradox remains points towards a new principle to resolve extended Wigner's friend paradoxes. 
For each assumption, we state the general principle and specify its application to the GHZ--FR paradox.

\subsection{GHZ--FR truth no-go theorem} \label{sec:GHZ_nogo}
\subsubsection{Assumptions} \label{sec:subsubsec_truth_nogo}
The name `the GHZ--FR \textit{truth} no-go theorem' refers to the assumption of an underlying `absolute truth' about the observations made in each \emph{single run} of the protocol, formalised as {\AOEfull} (\cref{assumption:AOE}) below
\cite{haddara2022possibilistic,cavalcanti2021implications,bong2020strong}.

We begin by presenting a list of assumptions that capture how we usually think about outcomes of performed measurements in experiments, namely the conjunction of {\possBorn}, {\AOEfull}, and {\BornCAOE}.
Once we have established this common ground, we add the assumption that superobservers exist and that they can be described as in \Cref{sec:Wignerfriend}, \cref{eq:unitary_friend} --, formalised as the assumption of {\universality}.

In quantum theory, the Born rule is used to obtain probability distributions for measurement outcomes.
In situations that involve multiple observers and superobservers,
one may expect the Born rule not to be applicable to all combinations of (actually performed) measurements.
For example, it would be too much to expect it to apply to the measurements of Wigner and his friend, as Wigner may even undo the friend's measurement.
In the present context however, only a minimal use of the Born rule is needed.
We only require the Born rule to hold in situations in which agents believe they could still test their predictions.
For example, in the GHZ--FR protocol Ursula makes predictions for Bob and Charlie, believing that
she could check their outcomes, \ie that she could still bring all relevant outcomes together.
Indeed, for all she knows, Valentina and Wigner might not have performed their measurements yet; in fact, Ursula need not know even about their existence and thinks she can still obtain the outcomes of Bob and Charlie, making it natural for her to formulate a prediction about $b,c$.
By contrast, one \textit{cannot} apply the Born rule to the joint outcomes of measurements performed by a superobserver-observer pair, for example to Valentina's and Bob's measurements, in a single use of the rule as nobody could ever reasonably believe to obtain such outcomes jointly; see also \Cref{fig:example_restricted_Born rule}.
The Born rule may be valid more broadly, but we only require it be applicable in these instances.

Moreover, we only need \emph{possibilistic} use of the Born rule.

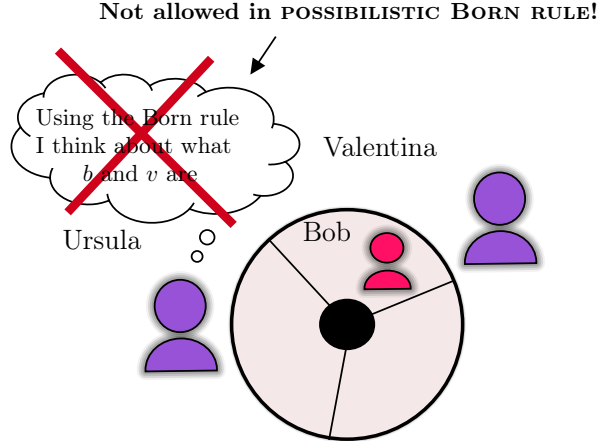
\begin{figure}
    \centering
    \tikzset{every picture/.style={line width=0.75pt}} 

\begin{tikzpicture}[x=0.75pt,y=0.75pt,yscale=-1,xscale=1]

\draw   (155.58,98.79) .. controls (156.9,91.93) and (152.56,85.15) .. (144.41,81.31) .. controls (136.26,77.47) and (125.72,77.26) .. (117.26,80.75) .. controls (114.27,76.77) and (108.79,74.02) .. (102.48,73.33) .. controls (96.17,72.65) and (89.77,74.11) .. (85.22,77.27) .. controls (82.67,73.66) and (77.67,71.23) .. (71.98,70.85) .. controls (66.29,70.47) and (60.72,72.19) .. (57.26,75.4) .. controls (52.65,71.57) and (45.32,69.96) .. (38.43,71.26) .. controls (31.55,72.57) and (26.35,76.55) .. (25.09,81.49) .. controls (19.44,82.58) and (14.74,85.34) .. (12.19,89.07) .. controls (9.65,92.8) and (9.51,97.13) .. (11.82,100.93) .. controls (6.25,106.04) and (4.96,112.85) .. (8.4,118.82) .. controls (11.84,124.78) and (19.51,129.01) .. (28.55,129.93) .. controls (28.62,135.53) and (32.96,140.67) .. (39.92,143.36) .. controls (46.88,146.06) and (55.36,145.89) .. (62.1,142.93) .. controls (64.97,149.63) and (73.04,154.57) .. (82.83,155.6) .. controls (92.62,156.63) and (102.38,153.57) .. (107.88,147.75) .. controls (114.62,150.62) and (122.71,151.45) .. (130.33,150.04) .. controls (137.95,148.64) and (144.44,145.12) .. (148.36,140.29) .. controls (155.25,140.85) and (161.91,138.33) .. (165.04,133.97) .. controls (168.16,129.61) and (167.09,124.34) .. (162.35,120.77) .. controls (168.5,118.22) and (171.63,113.15) .. (170.12,108.2) .. controls (168.62,103.26) and (162.8,99.57) .. (155.72,99.05) ; \draw   (162.35,120.77) .. controls (159.45,121.98) and (156.1,122.52) .. (152.74,122.34)(148.35,140.29) .. controls (146.91,140.17) and (145.5,139.91) .. (144.15,139.53)(107.88,147.75) .. controls (108.9,146.67) and (109.74,145.53) .. (110.41,144.32)(62.1,142.93) .. controls (61.58,141.7) and (61.24,140.44) .. (61.09,139.17)(28.55,129.93) .. controls (28.49,123.96) and (33.28,118.5) .. (40.88,115.89)(11.82,100.93) .. controls (13.05,102.96) and (14.93,104.76) .. (17.3,106.19)(25.09,81.49) .. controls (24.88,82.31) and (24.78,83.14) .. (24.8,83.98)(57.26,75.4) .. controls (58.41,76.36) and (59.35,77.43) .. (60.07,78.57)(85.22,77.27) .. controls (85.84,78.14) and (86.29,79.06) .. (86.59,80.01)(117.26,80.75) .. controls (115.48,81.5) and (113.82,82.39) .. (112.34,83.41)(155.58,98.79) .. controls (155.39,99.73) and (155.11,100.66) .. (154.72,101.58) ;
\draw  [color={rgb, 255:red, 0; green, 0; blue, 0 }  ,draw opacity=1 ][fill={rgb, 255:red, 151; green, 82; blue, 215 }  ,fill opacity=1 ][blur shadow={shadow xshift=0pt,shadow yshift=0pt, shadow blur radius=3.75pt, shadow blur steps=5 ,shadow opacity=100}] (88.04,203.49) .. controls (88.04,195.2) and (94.55,188.48) .. (102.57,188.48) .. controls (110.59,188.48) and (117.09,195.2) .. (117.09,203.49) .. controls (117.09,211.79) and (110.59,218.51) .. (102.57,218.51) .. controls (94.55,218.51) and (88.04,211.79) .. (88.04,203.49) -- cycle ;
\draw  [color={rgb, 255:red, 0; green, 0; blue, 0 }  ,draw opacity=1 ][fill={rgb, 255:red, 151; green, 82; blue, 215 }  ,fill opacity=1 ][blur shadow={shadow xshift=0pt,shadow yshift=0pt, shadow blur radius=3.75pt, shadow blur steps=5 ,shadow opacity=100}] (82.84,241.02) .. controls (82.83,240.82) and (82.82,240.63) .. (82.82,240.43) .. controls (82.68,230.48) and (91.67,222.28) .. (102.89,222.11) .. controls (114.12,221.94) and (123.34,229.87) .. (123.48,239.82) .. controls (123.48,240.04) and (123.48,240.26) .. (123.48,240.47) -- cycle ;
\draw  [fill={rgb, 255:red, 244; green, 232; blue, 232 }  ,fill opacity=1 ][line width=1.5] [blur shadow={shadow xshift=0pt,shadow yshift=0pt, shadow blur radius=1.5pt, shadow blur steps=4 ,shadow opacity=100}] (131.84,214.07) .. controls (131.84,177.72) and (161.27,148.26) .. (197.58,148.26) .. controls (233.89,148.26) and (263.32,177.72) .. (263.32,214.07) .. controls (263.32,250.41) and (233.89,279.87) .. (197.58,279.87) .. controls (161.27,279.87) and (131.84,250.41) .. (131.84,214.07) -- cycle ;
\draw   (113.8,164.09) .. controls (113.8,161.79) and (115.66,159.93) .. (117.96,159.93) .. controls (120.25,159.93) and (122.11,161.79) .. (122.11,164.09) .. controls (122.11,166.38) and (120.25,168.25) .. (117.96,168.25) .. controls (115.66,168.25) and (113.8,166.38) .. (113.8,164.09) -- cycle ;
\draw   (109.09,175.08) .. controls (109.09,173.36) and (110.48,171.96) .. (112.2,171.96) .. controls (113.93,171.96) and (115.32,173.36) .. (115.32,175.08) .. controls (115.32,176.81) and (113.93,178.2) .. (112.2,178.2) .. controls (110.48,178.2) and (109.09,176.81) .. (109.09,175.08) -- cycle ;
\draw  [color={rgb, 255:red, 0; green, 0; blue, 0 }  ,draw opacity=1 ][fill={rgb, 255:red, 151; green, 82; blue, 215 }  ,fill opacity=1 ][blur shadow={shadow xshift=0pt,shadow yshift=0pt, shadow blur radius=3.75pt, shadow blur steps=5 ,shadow opacity=100}] (269.83,141.94) .. controls (269.83,133.65) and (276.33,126.93) .. (284.35,126.93) .. controls (292.38,126.93) and (298.88,133.65) .. (298.88,141.94) .. controls (298.88,150.23) and (292.38,156.95) .. (284.35,156.95) .. controls (276.33,156.95) and (269.83,150.23) .. (269.83,141.94) -- cycle ;
\draw  [color={rgb, 255:red, 0; green, 0; blue, 0 }  ,draw opacity=1 ][fill={rgb, 255:red, 151; green, 82; blue, 215 }  ,fill opacity=1 ][blur shadow={shadow xshift=0pt,shadow yshift=0pt, shadow blur radius=3.75pt, shadow blur steps=5 ,shadow opacity=100}] (264.62,179.46) .. controls (264.61,179.27) and (264.61,179.08) .. (264.61,178.88) .. controls (264.46,168.93) and (273.45,160.73) .. (284.68,160.56) .. controls (295.91,160.39) and (305.12,168.32) .. (305.27,178.27) .. controls (305.27,178.49) and (305.27,178.7) .. (305.26,178.92) -- cycle ;
\draw  [fill={rgb, 255:red, 255; green, 16; blue, 101 }  ,fill opacity=1 ][blur shadow={shadow xshift=0pt,shadow yshift=0pt, shadow blur radius=3.75pt, shadow blur steps=5 ,shadow opacity=100}] (210.83,170.16) .. controls (210.83,165.48) and (215.03,161.69) .. (220.22,161.69) .. controls (225.41,161.69) and (229.61,165.48) .. (229.61,170.16) .. controls (229.61,174.83) and (225.41,178.63) .. (220.22,178.63) .. controls (215.03,178.63) and (210.83,174.83) .. (210.83,170.16) -- cycle ;
\draw  [fill={rgb, 255:red, 255; green, 16; blue, 101 }  ,fill opacity=1 ][blur shadow={shadow xshift=0pt,shadow yshift=0pt, shadow blur radius=3.75pt, shadow blur steps=5 ,shadow opacity=100}] (207.46,193.9) .. controls (207.45,193.77) and (207.45,193.65) .. (207.45,193.52) .. controls (207.35,187.91) and (213.17,183.28) .. (220.43,183.19) .. controls (227.69,183.09) and (233.65,187.56) .. (233.74,193.17) .. controls (233.75,193.31) and (233.74,193.45) .. (233.74,193.59) -- cycle ;
\draw  [fill={rgb, 255:red, 0; green, 0; blue, 0 }  ,fill opacity=1 ] (182.12,214.07) .. controls (182.12,206.2) and (189.04,199.83) .. (197.58,199.83) .. controls (206.12,199.83) and (213.04,206.2) .. (213.04,214.07) .. controls (213.04,221.93) and (206.12,228.31) .. (197.58,228.31) .. controls (189.04,228.31) and (182.12,221.93) .. (182.12,214.07) -- cycle ;
\draw    (152.83,165.19) -- (197.58,214.07) ;
\draw    (197.58,224.68) -- (187.28,280.26) ;
\draw    (197.58,214.07) -- (258.1,189.07) ;
\draw  [draw opacity=0][fill={rgb, 255:red, 255; green, 255; blue, 255 }  ,fill opacity=1 ] (25,115.8) -- (47.5,115.8) -- (47.5,124) -- (25,124) -- cycle ;
\draw    (157,49) -- (144.16,68.3) ;
\draw [shift={(142.5,70.8)}, rotate = 303.63] [fill={rgb, 255:red, 0; green, 0; blue, 0 }  ][line width=0.08]  [draw opacity=0] (8.93,-4.29) -- (0,0) -- (8.93,4.29) -- cycle    ;
\draw [color={rgb, 255:red, 208; green, 2; blue, 27 }  ,draw opacity=1 ][fill={rgb, 255:red, 208; green, 2; blue, 27 }  ,fill opacity=1 ][line width=3.75]    (33,57) -- (133,157) ;
\draw [color={rgb, 255:red, 208; green, 2; blue, 27 }  ,draw opacity=1 ][fill={rgb, 255:red, 208; green, 2; blue, 27 }  ,fill opacity=1 ][line width=3.75]    (120.22,61.69) -- (38.5,149.8) ;

\draw (34,158.18) node [anchor=north west][inner sep=0.75pt]  [font=\large] [align=left] {{\large Ursula}};
\draw (182.5,105.02) node [anchor=north west][inner sep=0.75pt]  [font=\large] [align=left] {{\large Valentina}};
\draw (170.8,153.57) node [anchor=north west][inner sep=0.75pt]  [font=\large] [align=left] {Bob};
\draw (19,89) node [anchor=north west][inner sep=0.75pt]   [align=left] {Using the Born rule \\I think about what \\$\displaystyle \ \ \ \ \ \ b$ and $\displaystyle v$ are};
\draw (55,28) node [anchor=north west][inner sep=0.75pt]   [align=left] {\textbf{Not allowed per Born Compatibility}};

\end{tikzpicture}
    \caption{An example of a reasoning situation following the GHZ--FR protocol as in \Cref{fig:GHZ_final} that is not allowed in {\protect\BornCAOE}: Valentina is a superobserver for Bob, performing a supermeasurement on Bob's lab. Ursula cannot use the Born rule to make a single prediction involving both Bob's and Valentina's outcomes $b$ and $v$, as she can never store both simultaneously in a memory.}
    \label{fig:example_restricted_Born rule}
\end{figure}

The simplest requirement regarding the possibilistic Born rule is to require its experimental validity, \ie that a classical agent can never experimentally disprove the possibilistic Born rule by gathering outcomes of measurements, possibly performed by other agents. We will refer to this as the principle \hypertarget{poss_Born}{\possBorn}.
We assume that the agent using the Born rule has, of course, a correct description of the measurements and states involved.

{\possBorn} allows an agent to make predictions about joint measurement outcomes which are indeed borne out if they gather the outcomes of all involved measurements.
The stronger requirement that such predictions made by an agent be \emph{valid} even if the agent does not actually bring those measurement outcomes together can be phrased as follows:
\begin{quote}
    Every performed measurement has an absolute, single-valued outcome about which the possibilistic Born rule as used by agents provides valid predictions.
\end{quote}
We can split this principle up into two philosophical principles:
the existence of an absolute truth, formalised in {\AOEfull},
and {\possBorn} as used by agents vis-à-vis this absolute truth,
regardless of whether its predictions are verified by gathering all involved outcomes, formalised in {\BornCAOE}.\footnotemark

\footnotetext{In epistemic logic,  {\AOEfull} could be formalised as the Truth axiom \cite{vilasini2019multi,nurgalieva2018inadequacy,montanhano2023contextuality}; see  \Cref{remark:epistemic} ahead for the formulation of key assumptions in our no-go results in the language of epistemic modal logic.}

\begin{assumptionp}{A1}[\hypertarget{AOE}{\AOEfull} ({\AOE})] \label{assumption:AOE}
Every performed measurement has an absolute, single-valued outcome. 

Concretely, in each run of the protocol, there exists an assignment $f: \mathcal{M} \rightarrow O$,
where $\mathcal{M}$ is the (finite) set of all the measurements performed by agents in the protocol,
mapping each performed measurement to its observed outcome.\footnotemark
\end{assumptionp}

\footnotetext{For simplicity, we take all measurements to be valued on the same outcome set $O$. If each measurement $m$ has its own set $O_M$ of (a priori) possible outcomes, then the type of $f$ would more properly be that of a dependent function, or a tuple, $f \in \prod_{m \in \mathcal{M}}O_m$, but we might as well take $O = \cup_{m \in \mathcal{M}} O_m$.}

We now require the absolute outcomes of {\AOE} to be compatible with an agent's use of the Born rule, making the agent's statements obtained from their application of the Born rule true regardless of whether the agent actually brings the involved outcomes together, but only when the agent believes to be able to gather the relevant outcomes in a joint memory.

\begin{assumptionp}{A2}[\hypertarget{Born_Compat_AOE}{\BornCAOE}]\label{assumption:borncompatibility}
    Outcomes assigned per {\AOE} to a set of measurements that an agent $A$ believes she can gather jointly in a memory must be compatible with $A$'s use of the possibilistic Born rule.
    \newline 
    Concretely, the assignment $f: \mathcal{M} \rightarrow O$ per {\AOE} is compatible with all allowed uses of the {\possBorn} by agents. 
\end{assumptionp}

Using {\AOEfull} and {\BornCAOE} an agent may translate their predictions from {\possBorn} to conditions on $f: \mathcal{M} \rightarrow O$.
For example, if the Born rule deems impossible (\ie assigns probability zero to) the joint outcome $(a_1,\ldots,a_n)$ of the measurements $M_1,\ldots,M_n$, then in each run of the protocol the `observed outcome' assignment $f$ cannot satisfy $f(M_i)=a_i$ for all $i=1,\ldots,n$.

Without any further restriction, the {\AOE} assumption is mathematically trivial, \ie a value assignment $f:\mathcal{M}\rightarrow O$ can always be found in each round of a measurement protocol such as the GHZ--FR.
However, {\BornCAOE} introduces restrictions on such value assignments, forcing them to conform with testable predictions, thereby rendering the underlying absoluteness requirement non-trivial.

\begin{remark}[{\AOE} and noncontextuality]
    The assumption {\AOEfull} closely resembles that of noncontextuality in the sheaf-theoretic framework \cite{abramsky2011sheaf}.
    In FR-like paradoxes the different contexts required for a contradiction are operationally brought together through the statements made by different agents.
    Physically, the difference with respect to Bell nonlocality scenarios lies in the fact that in FR-like paradoxes all the measurements that are assigned outcomes are actually performed in each single round of the protocol.
    The assumptions {\AOEfull} and {\BornCAOE} do not yield contradictions in quantum experiments without superobservers,\footnotemark 
    such as standard Bell or contextuality scenarios.
    In such setups, after each run of the experiment there exists a consistent assignment of outcome values to all the measurements that \textit{have been performed} in that run. 
    Therefore, {\AOEfull} does not involve any counterfactual statements about the outcomes of unperformed measurements, as in Bell locality or noncontextuality assumptions, since all involved measurements are actually performed by some agent.
    The only `counterfactual' aspect here is captured by the fact that the {\possBorn} is required to be valid regardless of whether the outcomes of the performed measurements are actually brought together, as per {\BornCAOE}. 
\end{remark}

\footnotetext{Even with a single observer-superobserver pair as in Wigner's original thought experiment, these two assumptions are not contradictory, as after each single run of the experiment there always exists a consistent history of outcome assignments.}
    
The two \cref{assumption:AOE,assumption:borncompatibility} capture how we normally reason in (quantum) experiments.
Our final assumption introduces superobservers.

\begin{assumptionp}{A3}[\hypertarget{Universality}{\universality}]\label{assumption:universality}
An observer can perform a measurement on a quantum system. Furthermore, a sealed lab in which an observer performs a measurement on some system may be described as a closed quantum system that evolves unitarily, upon which a superobserver can apply quantum operations, including measurements.\\
Concretely, this assumption ensures that the GHZ--FR protocol including superobservers (pictured in circuit form in \Cref{fig:GHZ_final}, with the parties situated as in \Cref{fig:FR_GHZ_lightcone}) can be performed. 
More specifically: first, there exists a GHZ state, upon which Alice, Bob, and Charlie can perform their $Y$-measurements; these measurements and everything else happening in each of the observers' sealed labs can be modelled as a known unitary by a superobserver; secondly, superobservers can perform supermeasurements on observers, that is, Ursula, Valentina, and Wigner can perform the $X$-like supermeasurements described in \Cref{sec:GHZ_protocol}. 
\end{assumptionp}

We are now in a position to state the GHZ--FR truth no-go theorem.

\begin{theorem}[GHZ--FR truth no-go theorem] \label{th:GHZ_FR_no_go}
   The GHZ--FR paradox shows that the  {\universality}, {\AOEfull} and {\BornCAOE} are incompatible.
\end{theorem}
\begin{proof} 
    The assumption {\universality} ensures that the GHZ--FR measurement protocol can be performed. 
    For the reasoning stage of the paradox, the following assumptions are required. 
    In each run of the protocol, by {\AOEfull}, there exists a function $f \colon \mathcal{M}\to O$ assigning an outcome to each performed measurement.
    Writing $A, B, C, U, V, W$ for the measurements performed in the GHZ--FR paradox, such assignment $f$ consists of outcome values $f(A)$, $f(B)$, $f(C)$, $f(U)$, $f(V)$, $f(W)$, which were denoted by $a,b,c,u,v,w$, respectively, in \Cref{sec:GHZ_protocol}.

    Using {\BornCAOE} (which can be applied since Ursula has no knowledge about Valentina and Wigner and thus believes she is able to gather $u,b,c$ jointly), Ursula reasons that this assignment must satisfy \begin{equation} f(B) \oplus f(C) = 1 \oplus f(U), \end{equation} and similarly Valentina and Wigner respectively reason that
    \begin{align} f(A)  \oplus f(C) &= 1\oplus f(V), \\ f(A) \oplus f(B) &= 1 \oplus f(W). \end{align}
    Finally, applying {\possBorn} (implied by {\BornCAOE}) to the classically recorded outcomes of $U,V,W$, which are available after the protocol, we know that these must satisfy \begin{equation}
        f(U) \oplus f(V) \oplus f(W)=0.
    \end{equation}
    But no assignment $f$ can simultaneously satisfy the four required equations, leading to a contradiction.
\end{proof}

\subsubsection{Ways to resolve the truth no-go theorem} 
\label{sec:ways_to_resolve_truth}

The assumptions of the GHZ--FR truth no-go theorem (\Cref{th:GHZ_FR_no_go}) are shown in \Cref{fig:GHZ_FR_assumptions}.
The theorem forces us to reject at least one of these three assumptions. We discuss each of the four possibilities in turn.

\begin{figure}[H]
\centering
\tikzset{every picture/.style={line width=0.75pt}} 

\begin{tikzpicture}[x=0.75pt,y=0.75pt,yscale=-1,xscale=1]

\draw    (641.5,186) -- (651.16,174.56) -- (665.56,157.49) ;
\draw [shift={(667.5,155.2)}, rotate = 130.17] [fill={rgb, 255:red, 0; green, 0; blue, 0 }  ][line width=0.08]  [draw opacity=0] (8.93,-4.29) -- (0,0) -- (8.93,4.29) -- cycle    ;
\draw    (243.94,212.79) -- (186,159.22) ;
\draw [shift={(183.8,157.19)}, rotate = 42.75] [fill={rgb, 255:red, 0; green, 0; blue, 0 }  ][line width=0.08]  [draw opacity=0] (8.93,-4.29) -- (0,0) -- (8.93,4.29) -- cycle    ;
\draw    (164.26,130.33) -- (164.26,73.6) ;
\draw [shift={(164.26,70.6)}, rotate = 90] [fill={rgb, 255:red, 0; green, 0; blue, 0 }  ][line width=0.08]  [draw opacity=0] (8.93,-4.29) -- (0,0) -- (8.93,4.29) -- cycle    ;
\draw    (243.94,212.79) -- (243.44,74.64) ;
\draw [shift={(243.43,71.64)}, rotate = 89.79] [fill={rgb, 255:red, 0; green, 0; blue, 0 }  ][line width=0.08]  [draw opacity=0] (8.93,-4.29) -- (0,0) -- (8.93,4.29) -- cycle    ;
\draw    (19.19,130.38) -- (19.42,73.6) ;
\draw [shift={(19.43,70.6)}, rotate = 90.23] [fill={rgb, 255:red, 0; green, 0; blue, 0 }  ][line width=0.08]  [draw opacity=0] (8.93,-4.29) -- (0,0) -- (8.93,4.29) -- cycle    ;
\draw  [fill={rgb, 255:red, 0; green, 0; blue, 0 }  ,fill opacity=1 ] (-7.5,30) -- (271.5,30) -- (271.5,70) -- (-7.5,70) -- cycle ;
\draw  [fill={rgb, 255:red, 46; green, 90; blue, 140 }  ,fill opacity=1 ] (-42.5,103.57) -- (80.87,103.57) -- (80.87,157.19) -- (-42.5,157.19) -- cycle ;
\draw  [fill={rgb, 255:red, 46; green, 90; blue, 140 }  ,fill opacity=1 ] (101.55,103.52) -- (226.98,103.52) -- (226.98,157.14) -- (101.55,157.14) -- cycle ;
\draw  [fill={rgb, 255:red, 46; green, 90; blue, 140 }  ,fill opacity=1 ] (176.85,185.98) -- (310.5,185.98) -- (310.5,239.6) -- (176.85,239.6) -- cycle ;
\draw    (668.5,135.2) -- (668.5,72.2) ;
\draw [shift={(668.5,69.2)}, rotate = 90] [fill={rgb, 255:red, 0; green, 0; blue, 0 }  ][line width=0.08]  [draw opacity=0] (8.93,-4.29) -- (0,0) -- (8.93,4.29) -- cycle    ;
\draw    (626.5,206) -- (580.05,158.64) ;
\draw [shift={(577.95,156.5)}, rotate = 45.56] [fill={rgb, 255:red, 0; green, 0; blue, 0 }  ][line width=0.08]  [draw opacity=0] (8.93,-4.29) -- (0,0) -- (8.93,4.29) -- cycle    ;
\draw    (561.25,127.55) -- (561.25,71.61) ;
\draw [shift={(561.25,68.61)}, rotate = 90] [fill={rgb, 255:red, 0; green, 0; blue, 0 }  ][line width=0.08]  [draw opacity=0] (8.93,-4.29) -- (0,0) -- (8.93,4.29) -- cycle    ;
\draw    (637.62,210.3) -- (637.11,71.66) ;
\draw [shift={(637.1,68.66)}, rotate = 89.79] [fill={rgb, 255:red, 0; green, 0; blue, 0 }  ][line width=0.08]  [draw opacity=0] (8.93,-4.29) -- (0,0) -- (8.93,4.29) -- cycle    ;
\draw    (412.72,127.55) -- (412.96,71.4) ;
\draw [shift={(412.97,68.4)}, rotate = 90.24] [fill={rgb, 255:red, 0; green, 0; blue, 0 }  ][line width=0.08]  [draw opacity=0] (8.93,-4.29) -- (0,0) -- (8.93,4.29) -- cycle    ;
\draw  [fill={rgb, 255:red, 0; green, 0; blue, 0 }  ,fill opacity=1 ] (384.5,28) -- (708.5,28) -- (708.5,68) -- (384.5,68) -- cycle ;
\draw  [fill={rgb, 255:red, 46; green, 90; blue, 140 }  ,fill opacity=1 ] (350.5,100.64) -- (474.95,100.64) -- (474.95,154.45) -- (350.5,154.45) -- cycle ;
\draw  [fill={rgb, 255:red, 46; green, 90; blue, 140 }  ,fill opacity=1 ] (497.99,100.64) -- (624.51,100.64) -- (624.51,154.45) -- (497.99,154.45) -- cycle ;
\draw  [fill={rgb, 255:red, 189; green, 16; blue, 224 }  ,fill opacity=1 ] (571.04,185.12) -- (648.5,185.12) -- (648.5,239.2) -- (571.04,239.2) -- cycle ;
\draw  [fill={rgb, 255:red, 189; green, 16; blue, 224 }  ,fill opacity=1 ] (645.96,101.34) -- (724.5,101.34) -- (724.5,155.43) -- (645.96,155.43) -- cycle ;

\draw (38.37,39.63) node [anchor=north west][inner sep=0.75pt]  [font=\normalsize,color={rgb, 255:red, 255; green, 255; blue, 255 }  ,opacity=1 ] [align=left] {GHZ--FR truth no-go theorem};
\draw (190.92,197.94) node [anchor=north west][inner sep=0.75pt]  [font=\normalsize,color={rgb, 255:red, 255; green, 255; blue, 255 }  ,opacity=1 ] [align=left] { Absoluteness of \\Observed \ Events};
\draw (100.14,119.58) node [anchor=north west][inner sep=0.75pt]  [font=\normalsize,color={rgb, 255:red, 255; green, 255; blue, 255 }  ,opacity=1 ] [align=left] { Born Compatibility};
\draw (-35.08,112.52) node [anchor=north west][inner sep=0.75pt]  [font=\normalsize,color={rgb, 255:red, 255; green, 255; blue, 255 }  ,opacity=1 ] [align=left] { \ Universality of \\Quantum Theory};
\draw (-49,-5) node [anchor=north west][inner sep=0.75pt]   [align=left] {(a)};
\draw (380,-2) node [anchor=north west][inner sep=0.75pt]   [align=left] {(b)};
\draw (439.37,38.63) node [anchor=north west][inner sep=0.75pt]  [font=\normalsize,color={rgb, 255:red, 255; green, 255; blue, 255 }  ,opacity=1 ] [align=left] {GHZ--FR agreement no-go theorem};
\draw (497.13,119.58) node [anchor=north west][inner sep=0.75pt]  [font=\normalsize,color={rgb, 255:red, 255; green, 255; blue, 255 }  ,opacity=1 ] [align=left] { Born Compatibility};
\draw (357.53,110.75) node [anchor=north west][inner sep=0.75pt]  [font=\normalsize,color={rgb, 255:red, 255; green, 255; blue, 255 }  ,opacity=1 ] [align=left] { \ Universality of \\Quantum Theory};
\draw (575.19,193.45) node [anchor=north west][inner sep=0.75pt]  [font=\normalsize,color={rgb, 255:red, 255; green, 255; blue, 255 }  ,opacity=1 ] [align=left] { Personal \\Knowledge};
\draw (654.17,109.45) node [anchor=north west][inner sep=0.75pt]  [font=\normalsize,color={rgb, 255:red, 255; green, 255; blue, 255 }  ,opacity=1 ] [align=left] { Classical \\Agreement};

\end{tikzpicture}  \caption{Assumptions leading to (a) the GHZ--FR truth no-go theorem of \Cref{sec:GHZ_nogo} and (b) the GHZ--FR agreement no-go theorem of \Cref{sec:GHZ_nogo_trust}, where AOE is replaced by Personal Knowledge and Classical Agreement. An arrow from one assumption box to another denotes that the former is used to formulate the latter.} \label{fig:GHZ_FR_assumptions} 
\end{figure}

\paragraph*{a. Reject {\AOEfull}}
Refuting {\AOEfull} means one cannot simply talk about single-valued outcomes obtained in measurements performed by other agents as being absolute, \ie a function $f:\mathcal{M} \rightarrow O$ assigning an outcome to each performed measurement cannot exist. Still, we would like agents to be able to formulate predictions correctly. 
Therefore, we may ask `When can an agent assign a single-valued outcome to a performed measurement?'
Naturally, the experimenter who performs the measurement can assign a single-valued outcome, namely the one they observe. 
This extends also to any agent who learns about the outcome, \ie any agent obtaining a copy of the outcome record and storing it in their memory\footnote{For a more elaborate definition of \emph{learning about an outcome}, see \Cref{appendix:clarific}.}. 
Can other agents assign an outcome -- that is not a priori absolute -- too?
We probe this possibility in the GHZ--FR agreement no-go theorem in the next section, where we weaken {\AOEfull} to allow for each agent to assign a personal single-valued outcome to performed measurements, which may a priori differ among different agents. The stronger no-go result obtained indicates that this approach may not lead to a resolution,
and that perhaps only the experimenter and agents learning about the outcome can assign an outcome value to a performed experiment.

\paragraph*{b. Reject {\BornCAOE}} As {\BornCAOE} assumes {\AOE}, refuting {\AOE} implies refuting {\BornCAOE}.
However, one may also consider accepting {\AOE} but refuting {\BornCAOE}.
In that case, there is still an absolute single-valued outcome assigned to each performed measurement,
but the use of the Born rule per {\possBorn} does not provide valid predictions about such assignment. 
In other words, more cautious use of the Born rule is required. Therefore, one needs an additional prescription for how all agents can use the Born rule. This is the route typically taken by (contextual) hidden variable theories. In this case, the observers' memories may also be described by hidden variables, such as in Bohmian theories \cite{sudbery2017single,lazarovici2019quantum,sudbery2019hidden}, making correct use of the Born rule more subtle.
We will discuss this option further in  \Cref{sec:GHZ_nogo_trust}.

When refuting {\BornCAOE}, one can instead go further as well and refute the experimental validity of the Born rule ({\possBorn}) as well, \ie modify the Born rule for all applications rather than only the non-experimental ones.

\paragraph*{c. Reject {\universality}} 
To reject the existence of superobservers, however, comes with its own complications.
If we adopt the viewpoint that only classical observers can perform measurements, where do we draw the line?
Would a hypothetical quantum computer showing signs of artificial intelligence not qualify as an observer? 
Or would the notion of a classical observer arise as an emergent concept?
Alternatively, one may reject the notion of measurement altogether.

In the foundations of quantum theory, the above questions deal with where to place the Heisenberg cut -- or whether it even exists.
In fact, Wigner's original thought experiment probed precisely this question.
However, with the FR and GHZ--FR paradoxes, the problem gains new importance, due to the stronger implications of accepting superobservers.
Namely, seemingly innocent use of quantum theory leads to contradictory statements regarding not only the question of \textit{when} a measurement does happen and classical outcomes are produced, but also regarding  \textit{what} actual outcome values were observed. 
In short, inconsistent histories may be obtained by different observers.

\medskip

In other words, if one accepts the experimental validity of the (possibilistic) Born rule and the existence of superobservers {\ref{assumption:universality}}, the contradiction in the GHZ--FR paradox arises through the assumed existence of a function $f\colon\mathcal{M} \rightarrow O$ assigning absolute, single-valued outcomes to all performed measurements which are compatible with the agents' allowed uses of the Born rule. 
In that case,  one must refute either the existence of such an absolute assignment $f$, thus rejecting {\AOE} {\ref{assumption:AOE}}, or the compatibility of $f$ with the agents' use of the Born rule, rejecting {\BornCAOE} {\ref{assumption:borncompatibility}} and require more cautious use of the Born rule.

We refer to \Cref{sec:comparison_FR} for a discussion of the FR assumptions, responses in literature and their relevance for the GHZ--FR no-go theorems.

\subsection{GHZ--FR agreement no-go theorem} \label{sec:GHZ_nogo_trust}
\subsubsection{Assumptions} \label{sec:subssubsec_trust_nogo}
We now replace
{\AOEfull} by a weaker condition, split up into two assumptions.
We introduce {\PersK}, which a priori allows different agents to assign different single-valued outcomes to the same measurement,
but subject it to a consistency condition called {\ClassicalT}, which stipulates that a statement known to be true by a classical agent can be taken to be true by any other classical agent to whom it is communicated.
As in the truth no-go theorem,
we require the outcome assignments per {\PersK} to be compatible with the Born rule. 
These assumptions allow superobservers Ursula, Valentina and Wigner to formulate their predictions about outcomes and share them with each other, obtaining a contradiction.

We stress that only classical agents need to reason and communicate;
there is no need for classical agents to trust the reasoning of quantum agents as in \cite{frauchiger2018quantum,nurgalieva2018inadequacy,montanhano2023contextuality}; see also \Cref{remark:epistemic} below.

\begin{assumptionp}{A1$^\prime$}[\hypertarget{Pers_K}{\PersK}]\label{assumption:personalknowledge}
Every performed measurement can be assigned a single-valued outcome by an agent who believes they are able to gather that outcome.\footnote{For example, in the GHZ--FR paradox Ursula applies {\PersK} to assign values to Bob's and Charlie's measurements, whose outcomes she believes she could still obtain as she does not need to know about Valentina's or Wigner's existence or operations.}
\newline
Concretely, in each run of the protocol, for each agent $A$, there exists an assignment $f_A: \mathcal{M}_A \rightarrow O$,
where $\mathcal{M}_A$ is the set of performed measurements whose outcomes could be gathered from $A$'s perspective,
mapping each performed measurement to an outcome value assigned by $A$.
\end{assumptionp}
The agents (superobservers) assigning outcomes in the GHZ--FR paradox only do so \emph{after} having performed their own measurement, \ie based on what they believe is available to them then.

We recall that {\BornCAOE} builds on {\AOE}, but we can also state an analogous Born rule compatibility requirement for this weakened version of {\AOE}.

\begin{assumptionp}{A2$^\prime$(i)}[\hypertarget{Born_Compat_PersK}{\BornCPersK}]\label{assumption:borncompatibility_truth}
Outcomes assigned \textit{by an agent} $A$ per {\PersK} to a set of measurements must be compatible with $A$'s use of the possibilistic Born rule.
\newline
Concretely, for every agent $A$, the assignment $f_A \colon \mathcal{M}_A \rightarrow O$ per {\PersK} is compatible with all allowed uses of the {\possBorn} by $A$.
\end{assumptionp}

In {\PersK}, each function $f_A$ assigns outcomes to measurements, but, crucially, these outcomes need not be absolute but can be relative to agent $A$. 
Using {\BornCPersK}, agent $A$ may translate her predictions from {\possBorn} to conditions on $f_A:\mathcal{M}_A \rightarrow O$. For example, if the Born rule deems impossible the joint outcome $(a_1,\ldots,a_n)$ of the measurements $M_1,\ldots, M_n$, then the assignment $f_A$ cannot satisfy $f_A(M_i) = a_i$ for all $i=1,\ldots,n$. Agent $A$ may conclude:  \begin{quote}
`I (agent $A$) know that there exists a measurement $M_i \in \{M_1,\ldots,M_n\}$, possibly performed by other agents, that does not have outcome $a_i$ (for me),
regardless of whether I ask the involved agent(s) for it.'
\end{quote}

Next, we state a minimal practical requirement on the outcome assignments across different classical agents.
This assumption was not needed for the GHZ--FR truth no-go theorem, as it is implied by {\AOE} and {\BornCAOE}.

\begin{assumptionp}{A2$^\prime$(ii)}[\hypertarget{Classical_T}{\ClassicalT}]\label{assumption:classicaltrust}
The outcome assignments (per \PersK) of classical agents who communicate classically must agree on their overlap.
\newline
Concretely, if classical agents $A$ and $B$ classically communicate their outcome assignments, then these assignments $f_A$ and $f_B$ must agree on their overlap $\mathcal{M}_A \cap \mathcal{M}_B$.
\end{assumptionp}

In fact, the way we use this assumption is that, for example, Ursula and Valentina, must agree on their outcome assignment for the measurement $C$: $f_U(C)=f_V(C)$. In fact, in a version of the protocol where Wigner performs his measurement later, both Ursula and Valentina can ask Charlie for his outcome value. They will obtain the same value, motivating this assumption. In our particular protocol, {\PersK} and {\ClassicalT} will lead to agents $U,V,W$ having to agree on their assignments so that again a single assignment $f$ to all performed measurements can be built that must satisfy the GHZ--Mermin equations.

\begin{theorem}[GHZ--FR agreement no-go theorem] \label{th:GHZ_FR_no_go_trust}
   The GHZ--FR paradox shows that the assumptions  {\universality},  {\PersK}, {\ClassicalT}, and {\BornCPersK}  are incompatible.
\end{theorem}
\begin{proof}
    The first part of the proof is analogous to the proof of \Cref{th:GHZ_FR_no_go}, replacing the absolute outcome assignment $f$ by the outcome assignments $f_U,f_V,f_W$ of Ursula, Valentina and Wigner per {\PersK}. For example, as Ursula has no knowledge about Valentina and Wigner and thus believes she is able to gather $u,b,c$ jointly, she can assign them outcome values $f_U(B),f_U(C)$ (and $f_U(U)$ of course), compatible with her Born rule use per {\PersK} and {\BornCPersK}. Then {\ClassicalT} requires these outcome assignments to agree on their overlaps as Ursula, Valentina and Wigner are treated classically in the protocol and can physically meet and communicate as much as they like, so that a single global outcome assignment can be built that satisfies the GHZ--Mermin equations. Details can be found in \Cref{appendix:proof_agreement_nogo}.
\end{proof}

\begin{remark}[Heisenberg cuts]
    In literature, sometimes $A$ applying the Born rule to set of jointly performed measurements is also referred to as $A$ choosing a Heisenberg cut, having the performed measurements outside the cut. In that terminology, our assumptions of {\PersK} and {\BornCPersK} can also be phrased as: `An agent $A$ can choose a Heisenberg cut, placing a set of measurements she thinks she can still gather jointly all outside the cut.' The paradox then arises if it is understood that different choices of Heisenberg cuts, when chosen by classical agents, must not lead to contradictions about the actual outcome values, which is another way of phrasing {\ClassicalT}. 
    Interpretations of quantum mechanics which resolve the measurement problem by determining or eliminating the Heisenberg cut will resolve EWF arguments by providing an unambiguous description for the entire protocol.
   We discuss Heisenberg cuts in more detail in \Cref{sec:previous_work}.
\end{remark}

\begin{remark}[The GHZ--FR assumptions in  modal logic] \label{remark:epistemic}
    We briefly outline a discussion of the GHZ--FR paradox in terms of epistemic logic, and refer to \Cref{appendix:epistemic_logic} for more details.

    Syntactically, in epistemic logic \cite{fagin2004reasoning,blackburn2001modal}, besides the usual logical connectives, formulas can be built using a unary modal operator $K_A$ for each agent $A$, with a formula $K_A \phi$ interpreted as `agent $A$ believes or knows that $\phi$ is True'. 
    Different sets of axioms are typically considered in epistemic logic systems, varying with the intended interpretation.
    
    In Ref.~\cite{nurgalieva2018inadequacy}, the \textit{Distribution Axiom} of modal logic (a.k.a.\ axiom \textbf{K}) is listed as assumption (D) of the FR paradox
    (see \Cref{sec:comparison_FR}), and
    used to justify the use of \textit{modus ponens} by quantum agents in the FR paradox. 
    In the GHZ--FR paradox, however, as in the version of the FR paradox from \Cref{sec:stricter_FR}, only classical agents need to reason, 
    dispensing with the reasoning by quantum agents considered in the earlier treatments of the FR paradox \cite{nurgalieva2018inadequacy,vilasini2019multi,montanhano2023contextuality,frauchiger2018quantum}.
    The fact that we need only consider reasoning by classical agents in the protocol is the reason why we do not state such reasoning as a separate assumption in our no-go theorems.
    In modal logical terms, this is reflected in the fact that we do not need the epistemic modalities $K_A$ for quantum agents:
    the whole argument can be phrased in a logical system with modalities only for each \textit{classical} agent in the protocol. 
    It could perhaps be argued that this somewhat weakens the appeal of a formalisation of the paradoxes using modal logic.

    Still, the core assumptions of our GHZ--FR no-go theorems can be phrased in these terms.
    The existence of an absolute truth, per \AOEfull, in the GHZ--FR truth no-go theorem corresponds to the Truth axiom: if an agent \textit{knows} something, then it is true: $K_A \phi \rightarrow \phi$. 
    The agreement no-go theorem requires a form of agreement only among classical agents.
    It can be seen as a considerable weakening of the epistemic trust structures identified for the FR paradox in \cite{nurgalieva2018inadequacy,vilasini2019multi,montanhano2023contextuality}, in that we do not need classical agents to believe the reasoning of quantum agents.
\end{remark}

\subsubsection{Ways to resolve the agreement no-go theorem} \label{sec:ways_to_resolve_trust}
The truth and agreement no-go theorems differ only in that the {\AOE} assumption in the truth no-go theorem is replaced by {\PersK} and {\ClassicalT} in the agreement no-go theorem (cf. \Cref{fig:GHZ_FR_assumptions}).
If one opts to resolve the truth no-go theorem by rejecting {\AOE}, then one must reject one of the assumptions {\PersK} and {\ClassicalT} to resolve the agreement no-go theorem.
We argue that, in that case, the most practical resolution is to further weaken {\PersK}.
Other resolutions are of course possible.

We now discuss the repercussions of separately rejecting each of the assumptions of the agreement no-go theorem. The consequences of rejecting {\universality} are the same as in the truth no-go theorem, as discussed in \Cref{sec:ways_to_resolve_truth}.

\paragraph*{a. Reject {\PersK}}\label{para:reject_Pers_Knowledge} 
Rejecting {\PersK} implies that one also rejects {\AOEfull} of the truth no-go theorem.
If the outcome of a performed measurement is not absolute, one must ask: `Who can assign a single-valued outcome to a performed measurement?'\footnote{As the outcome of a performed measurement is considered to be classical information, a physical entity assigning an outcome to a measurement
needs a preferred basis in its quantum storage system, as argued in \cite{brukner2021qubits,allam2023observer}, which consitutes our definition of an agent.}
Certainly an experimenter can assign a single value to the outcome of a measurement performed by themself and test the Born rule in this way -- but who else?

Suppose that an agent $A$'s prediction about outcomes of performed measurements is valid `no matter whether $A$ actually brings the involved outcomes together'. Then, one is assuming the existence of absolute outcomes or outcomes relative to the agent $A$ as in {\PersK}.
Therefore, rejecting {\AOE} and {\PersK} suggests that only the experimenter themself can assign a single-valued outcome to a performed measurement, and by extension any other agent learning about the outcome.
This proposal suggests a new principle: 
 {\ROE}\footnote{
 The seed of this principle originated in a discussion of L.W. with E. Cavalcanti about the implications of the Local Friends no-go theorem \cite{bong2020strong,cavalcanti2021implications}, but it is not necessarily what E. G. Cavalcanti believes to be a resolution.}:
\begin{quote}
    \hypertarget{ROE}{\ROE}
    Every performed measurement can be assigned a single-valued outcome by an agent who learns about that outcome.
    \end{quote}

Consequently, an agent $A$ can reason about the outcome $b$ of a measurement performed by another agent $B$ only to make statements about the value of this outcome $b$ \textit{conditioned on} agent $A$ asking $B$ for it, formalised in the following principle:
\begin{quote}
\hypertarget{BornP}{\BornP} An agent $A$ can translate predictions based on the Born rule only into statements about outcomes that $A$ would obtain when $A$ learns about these outcomes.

\end{quote}
To \emph{learn about an outcome}, an agent may (i) directly ask the experimenter, (ii) measure the measured system in the same measurement basis or (iii) indirectly obtain the measurement outcome through traces in the environment, be told about it by another agent, \ldots

\paragraph*{b. Reject {\BornCPersK}} \label{par:reject_Born_compatibility} If one rejects Born Compatibility, one can still keep {\AOE} (and {\PersK}).
However, in that case, as the GHZ--FR paradox shows, one needs a rule for when to expect Born rule predictions to be valid, as agents need to be able to phrase their Born rule predictions correctly.
When holding on to the experimental validity of the Born rule per {\possBorn}, as a principle for when such predictions are valid, one can use {\BornP} as well.
If an agent has more knowledge of the underlying theory or ontology their world is subject to, they can perhaps make more precise statements. 
For example, one can imagine a theory or ontology where AOE is satisfied (as in Bohmian mechanics), and where an agent can say their prediction is valid as long as they can gather all involved outcomes.
However, one can imagine a theory where additionally actually asking for the outcomes changes the actual outcomes through hidden variables. 
Therefore, without further knowledge, in theories rejecting {\BornCPersK}, agents can use {\BornP} for correct predictions.

\paragraph*{c. Reject {\ClassicalT}} \label{par:reject_classical_trust} Rejecting {\ClassicalT} would mean that classically communicating agents (communicating in whichever way they prefer) disagree on the truth value of a statement about what happened in a measurement.
In principle, every performed measurement whose outcome is not stored in some classical agent's memory can be part of a paradox à la GHZ--FR. Thus, disagreement among classically communicating agents about what happened in the measurement $M$ topples a basic assumption of science: we cannot trust other scientists' reports of what happened in an experiment.
To see this difficulty in more detail, we consider the following situation. Alice performs a measurement, and Bob makes a statement about Alice's outcome, claiming for example that Alice obtained $a=0$. Bob could have arrived at this statement by reasoning or by explicitly asking Alice. 
Holding on to {\PersK} and {\BornCPersK}, there is no need for Bob to specify how he obtained the statement, which is true \textit{for him}. 
Therefore, when Bob tells us that Alice obtained $a=0$, we simply cannot know whether his statement (or any statement derived from it) is true \textit{for us}. 
A potential way out is that Bob conveys more details about how he obtained his conclusion, and then we decide whether or not to believe him, as argued in Qbism \cite{debrota2020respecting}, for example.
In our everyday world, such disagreement between agents does not arise, because of the absence of superobservers.

\subsection{A possible resolution} \label{sec:possible_resolution}
In \Cref{sec:ways_to_resolve_trust} we proposed {\BornP} to resolve the GHZ--FR paradox, and by extension all extended Wigner's friend paradoxes, if one accepts the existence of superobservers. 
The principle says that observers should state their Born rule predictions \emph{conditioned} on learning about the involved outcomes.
In theories that reject {\AOE}, this principle can be implemented as {\ROE}, as explained in \Cref{para:reject_Pers_Knowledge}.
This resolution can be seen as an agent-based extension of Peres's dictum that `Unperformed experiments have no results' \cite{peres1978unperformed},
\begin{quote}
    Agent-based Peres's dictum:
    `Unperformed experiments have no results, and unknown results have no values.'
\end{quote}
We can see this as an extension of Peres's dictum, by considering learning about the outcome result (of a measurement performed by another agent) as an experiment too. 
Namely, you cannot \emph{a priori} assign outcome values to measurements or experiments, possibly performed by other agents, you can only do so upon actually learning the outcome. 
In a sense, this principle is a natural classical vs.\ quantum separation.
If a superobserver Wigner can model his friend's measurement unitarily, \ie as a quantum system, then Wigner cannot assign classical information to his friend's measurement outcome, as this information is still quantum for Wigner. 
Instead, Wigner should only assign a single value to his friend's outcome when he learns about that outcome.
Here is an example of {\BornP} in use for a simple Wigner's friend scenario:
\begin{quote}
    Wigner finds that $f=1$ is impossible (has probability zero) with $f$ denoting his friend's outcome and may conclude:
    ``Using the Born rule I (Wigner) know that my friend would never reply that she obtained $f=1$ \emph{if} I were to ask her for her outcome, (as long as my friend's memory record and mine are not wiped or updated).''
\end{quote}

Let us shortly explain how assuming {\BornP} instead of {\AOE} (or {\PersK}) and {\BornCAOE} resolves the GHZ--FR paradox. 
The FR paradox is resolved analogously, as presented in \Cref{appendix:appendix_resolution_FR}.

We assume that the GHZ--FR measurement protocol has been performed, considering a run with the (classically available) outcomes $u=0,v=0,w=0$, say.
Ursula, making her predictions, finds using {\BornP}:
\begin{quote}
    Ursula: `if I (Ursula) were to ask Bob and Charlie for their outcomes $b,c$, then I would find that $b \oplus c = 1$, as long as my memory and theirs are not wiped or updated.' 
\end{quote} 
Bob and Charlie argue similarly to conlude that $a \oplus c = 1$ and $a \oplus b = 1$ when they ask for $(a,c)$ and $(a,b)$, respectively. 
However, none of these superobservers have actually asked the other observers Alice, Bob and Charlie for their outcomes, so that no paradox arises.
For completeness, all precise statements are presented in \Cref{appendix:ROE_resolution}.

\section{Comparison with the FR paradox} \label{sec:comparison_FR}
We compare the GHZ--FR paradox to the FR paradox in \Cref{sec:comparison_FR_GHZ_FR}. 
Then, in \Cref{sec:previous_work}, we place our work within the broader literature on the subject: we briefly review previous responses to the FR paradox, as well as some work closely related to the GHZ--FR paradox, where similar measurement protocols have been proposed but with different analyses.

\subsection{Why the GHZ--FR paradox is stronger than the FR paradox} \label{sec:comparison_FR_GHZ_FR} 
To establish the GHZ--FR paradox, in stark contrast to the FR paradox \cite{frauchiger2018quantum,vilasini2022general}, (i) only \emph{classical} agents reason, and (ii) \emph{no post-selection} on outcomes is required.
The Born rule is only used to determine that certain (joint) measurement outcomes are \emph{impossible}.
It need not be used to assert that certain outcomes are \emph{possible}, that is, that they occur with (some) nonzero probability.
While the former are definitive statements about \emph{every} run of the experiment -- they state that a particular combination of outcomes is \emph{never} witnessed --, the latter inform about events that occur only in \emph{some} runs, which would need to be post-selected for the derived contradiction to apply.
This difference is at the root of distinctive feature (ii) above.

To compare the assumptions of the FR no-go theorem and those of the GHZ--FR agreement no-go theorem, the strongest of the two no-go theorems from \cref{sec:GHZ_FR_nogo}, we begin by listing the assumptions of both theorems (cf. \Cref{fig:assumptions_FR_GHZ_FR}).
For the FR theorem, we rely on the presentation in Ref.~\cite{vilasini2022general}, which
considers the following list of assumptions: (C) Consistency, (D) Distribution axiom of modal logic, (Q) Quantum Born rule statements, (U) Unitarity, and (S) Single-valuedness of outcomes.

Distinctive feature (i) -- that only classical agents reason -- is reflected in the weakening of assumptions (C), (D), and (Q) in the GHZ--FR paradox.
Distinctive feature (ii) -- that no post-selection is required -- is not reflected in the assumptions. 
It corresponds to the fact that the contradiction at the heart of the GHZ--FR paradox occurs in \emph{every} single round of the protocol,
whereas the FR paradox only obtains a contradiction in \emph{some} (postselected) rounds in which a specific value of the final (super)measurements is observed.

\paragraph*{The FR assumptions} We first summarise the assumptions of the FR no-go theorem in our own wording. All uses of the word `statement' are shorthand for `statement about outcome values observed by agents'.\footnote{In the FR paradox, one refers to outcome values observed by a single agent in these statements. In the GHZ--FR paradox, multiple agents are involved in such statements, for example Ursula reasons about $b,c$ observed by Bob and Charlie, respectively; Ursula treats these joint outcomes as an outcome of a joint measurement by Bob and Charlie. However, in the GHZ--FR no-go theorem we could also allow $A,B,C$ to share their outcome values with each other so that both $b,c$ are observed by Bob and Charlie for example, and have a superobserver later undo this sharing before other measurements are performed. We could even have all these measurements $A,B,C$ be performed by the same agent, with $U,V,W$ acting upon different parts of that agent's memory, so that when Ursula talks about outcomes $b,c$ this refers to a joint measurement performed by that single agent.}
\begin{itemize}
\item[(Q)] An agent $A$ can translate the use of the possibilistic quantum Born rule applied to a set of performed (not necessarily all by $A$) measurements, to make statements about what outcomes are obtained or not.
\item[(C)] 
If an agent $A$ knows that an agent $B$ knows that statement $s_1$ is true from reasoning with these same five rules, then agent $A$ knows that $s_1$ is true.
\item[(D)] If an agent $A$ knows a statement $s_1$ and also knows that $s_1$ implies statement $s_2$ then agent $A$ can conclude `I (A) know $s_2$'.\footnote{Assumption (D) refers to the Distribution axiom of Modal Logic and was first identified in \cite{nurgalieva2018inadequacy}, along with assumption (U).}
\item[(S)] Outcomes of performed measurements from an agent's viewpoint who uses (Q) are single-valued.
\item[(U)] An agent $A$ can model another agent $B$ measuring a system as a unitary process.
\end{itemize}
Let us give an example of (Q) in use. If agent $A$ applies the Born rule to a set of measurements that can be simulated in parallel and finds that a certain outcome $(a_1,\ldots,a_n)$ cannot occur, then $A$ can conclude that some $a_i$ is not obtained (in a measurement performed possibly by another agent).\footnote{We rephrased assumption (Q) here somewhat compared to FR's original phrasing, as unwanted interpretations of the latter may (wrongly) conclude collapse to be necessary for the FR paradox \cite{aaronson2018s,schmid2023review,tausk2018brief,bub2021understanding,mucino2020wigner}.}

\paragraph*{Rephrasing the GHZ--FR assumptions} We now rephrase the GHZ--FR assumptions in order to simplify the comparison with the assumptions of the FR paradox, \emph{emphasising} the differences.\footnote{In fact, this paraphrasing of our no-go theorem has slightly stronger assumptions than needed. We did so to tailor them more to the phrasing of the FR assumptions. To weaken them, we could allow the (Q) assumption to have statements about outcome values assigned by the agent making the prediction, subject to \ClassicalT.}
\begin{itemize}
    \item[(Q')] A \emph{classical}\footnote{We recall that by a \textit{classical} agent we mean an agent that need not be modelled quantumly by any other agent in the protocol.} agent $A$ can translate the use of the possibilistic quantum Born rule applied to a set of performed (not necessarily all by $A$) measurements, to make statements about what outcomes are obtained or not.
    \item[(C')] If a \textit{classical} agent $A$ knows that a \textit{classical} agent $B$ knows that statement $s_1$ is true from reasoning with the same five rules, \textit{and $B$ communicates his reasoning to $A$ through classical communication}, then agent $A$ knows that $s_1$ is true.
    \item[(D')] If a \emph{classical} agent $A$ knows a statement $s_1$ and also knows that $s_1$ implies statement $s_2$ then agent $A$ can conclude `I (A) know $s_2$'. 
\end{itemize}

\paragraph*{Comparing FR and GHZ--FR} Comparing the FR and GHZ--FR assumptions, (U) and (S) are the same, but the GHZ--FR assumptions (C'), (D'), (Q') are weaker versions of the FR assumptions (C), (D), (Q), respectively (cf. \Cref{fig:assumptions_FR_GHZ_FR}).
Namely, only \emph{classical agents} need to produce statements and reason, and accept statements classically communicated among them.\footnote{As described in \Cref{sec:stricter_FR}, a slight modification of the reasoning in the FR paradox can also produce a no-go theorem from these weaker assumptions. A similar modification of the reasoning in the FR paradox can be found in \cite{healey2018quantum}.} This is an important weakening, as for example the Sleeping Beauty problem of Ref.~\cite{jones2024thinking} shows how FR's assumption (C) might be problematic already in a classical setup. Their argument does not apply to its weakened version (C') in the GHZ--FR paradox.

\paragraph*{Explaining the rephrasing of the GHZ--FR assumptions}
We have argued that the GHZ--FR paradox requires assumptions (U), (S), (C'), (D'), and (Q'), modelled after those used in discussions of the FR paradox.
We now explain how they relate to the concepts used to describe the GHZ--FR agreement no-go theorem of \Cref{sec:GHZ_nogo_trust}; see also \Cref{fig:assumptions_FR_GHZ_FR}.
Concerning the truth no-go theorem, we recall that the assumptions {\ClassicalT} and {\PersK} of the agreement no-go theorem are replaced by the \textit{a priori} stronger {\AOEfull} in the truth no-go theorem. 
\begin{itemize}[]
\item[(U)] is the equivalent of the GHZ--FR assumption {\universality}. 
\item[(S)] is captured in {\PersK}, and {\AOEfull}, allowing an agent to assign a single-valued outcome to performed measurements. 
\item[(C')] is a (slightly stronger) paraphrasing of {\ClassicalT}. 

\item[(D')] allows for classical agents to reason and is assumed implicitly in the GHZ--FR agreement no-go theorem. Classical agents are just like us, using elementary rules of logic such as the \textit{modus ponens} as expressed by assumption (D').

\item[(Q')] is captured in the conjunction of {\AOE} or {\PersK} and {\BornCAOE}. 
The assumption (Q') assumes the validity of the Born rule. But in fact, (Q') needs more, as it allows an agent $A$ to make statements about outcomes obtained in a measurement performed by another agent $B$, regardless of whether $A$ tests this Born rule prediction or not, captured by {\PersK} and {\BornCPersK}, respectively. 
\end{itemize}

\begin{figure}[h]
\tikzset{every picture/.style={line width=0.75pt}} 

\begin{tikzpicture}[x=0.75pt,y=0.75pt,yscale=-1,xscale=1]

\draw  [fill={rgb, 255:red, 46; green, 90; blue, 140 }  ,fill opacity=1 ] (383.5,257) -- (603.5,257) -- (603.5,285) -- (383.5,285) -- cycle ;
\draw  [fill={rgb, 255:red, 46; green, 90; blue, 140 }  ,fill opacity=1 ] (382.5,97) -- (601.5,97) -- (601.5,125) -- (382.5,125) -- cycle ;
\draw  [color={rgb, 255:red, 208; green, 2; blue, 27 }  ,draw opacity=1 ][fill={rgb, 255:red, 225; green, 121; blue, 134 }  ,fill opacity=1 ] (84.5,112.2) .. controls (84.5,103.36) and (140.91,96.2) .. (210.5,96.2) .. controls (280.09,96.2) and (336.5,103.36) .. (336.5,112.2) .. controls (336.5,121.04) and (280.09,128.2) .. (210.5,128.2) .. controls (140.91,128.2) and (84.5,121.04) .. (84.5,112.2) -- cycle ;
\draw  [color={rgb, 255:red, 208; green, 2; blue, 27 }  ,draw opacity=1 ][fill={rgb, 255:red, 225; green, 121; blue, 134 }  ,fill opacity=1 ] (84.5,164.2) .. controls (84.5,155.92) and (141.36,149.2) .. (211.5,149.2) .. controls (281.64,149.2) and (338.5,155.92) .. (338.5,164.2) .. controls (338.5,172.48) and (281.64,179.2) .. (211.5,179.2) .. controls (141.36,179.2) and (84.5,172.48) .. (84.5,164.2) -- cycle ;
\draw  [color={rgb, 255:red, 208; green, 2; blue, 27 }  ,draw opacity=1 ][fill={rgb, 255:red, 225; green, 121; blue, 134 }  ,fill opacity=1 ] (84.5,217) .. controls (84.5,205.95) and (96.7,197) .. (111.75,197) .. controls (126.8,197) and (139,205.95) .. (139,217) .. controls (139,228.05) and (126.8,237) .. (111.75,237) .. controls (96.7,237) and (84.5,228.05) .. (84.5,217) -- cycle ;
\draw  [color={rgb, 255:red, 208; green, 2; blue, 27 }  ,draw opacity=1 ][fill={rgb, 255:red, 225; green, 121; blue, 134 }  ,fill opacity=1 ] (84.5,270) .. controls (84.5,258.95) and (96.7,250) .. (111.75,250) .. controls (126.8,250) and (139,258.95) .. (139,270) .. controls (139,281.05) and (126.8,290) .. (111.75,290) .. controls (96.7,290) and (84.5,281.05) .. (84.5,270) -- cycle ;
\draw  [color={rgb, 255:red, 208; green, 2; blue, 27 }  ,draw opacity=1 ][fill={rgb, 255:red, 225; green, 121; blue, 134 }  ,fill opacity=1 ] (84.5,323) .. controls (84.5,311.95) and (96.7,303) .. (111.75,303) .. controls (126.8,303) and (139,311.95) .. (139,323) .. controls (139,334.05) and (126.8,343) .. (111.75,343) .. controls (96.7,343) and (84.5,334.05) .. (84.5,323) -- cycle ;
\draw  [color={rgb, 255:red, 208; green, 2; blue, 27 }  ,draw opacity=1 ][fill={rgb, 255:red, 225; green, 121; blue, 134 }  ,fill opacity=1 ] (282.5,217) .. controls (282.5,205.95) and (294.7,197) .. (309.75,197) .. controls (324.8,197) and (337,205.95) .. (337,217) .. controls (337,228.05) and (324.8,237) .. (309.75,237) .. controls (294.7,237) and (282.5,228.05) .. (282.5,217) -- cycle ;
\draw  [color={rgb, 255:red, 208; green, 2; blue, 27 }  ,draw opacity=1 ][fill={rgb, 255:red, 225; green, 121; blue, 134 }  ,fill opacity=1 ] (282.5,270) .. controls (282.5,258.95) and (294.7,250) .. (309.75,250) .. controls (324.8,250) and (337,258.95) .. (337,270) .. controls (337,281.05) and (324.8,290) .. (309.75,290) .. controls (294.7,290) and (282.5,281.05) .. (282.5,270) -- cycle ;
\draw  [color={rgb, 255:red, 208; green, 2; blue, 27 }  ,draw opacity=1 ][fill={rgb, 255:red, 225; green, 121; blue, 134 }  ,fill opacity=1 ] (282.5,323) .. controls (282.5,311.95) and (294.7,303) .. (309.75,303) .. controls (324.8,303) and (337,311.95) .. (337,323) .. controls (337,334.05) and (324.8,343) .. (309.75,343) .. controls (294.7,343) and (282.5,334.05) .. (282.5,323) -- cycle ;
\draw  [fill={rgb, 255:red, 46; green, 90; blue, 140 }  ,fill opacity=1 ] (383.5,307) -- (603.5,307) -- (603.5,335) -- (383.5,335) -- cycle ;
\draw  [fill={rgb, 255:red, 46; green, 90; blue, 140 }  ,fill opacity=1 ] (383.5,144) -- (604.5,144) -- (604.5,172) -- (383.5,172) -- cycle ;
\draw  [fill={rgb, 255:red, 46; green, 90; blue, 140 }  ,fill opacity=1 ] (383.5,213) -- (604.5,213) -- (604.5,241) -- (383.5,241) -- cycle ;
\draw [color={rgb, 255:red, 74; green, 144; blue, 226 }  ,draw opacity=1 ]   (337,112) -- (376,112) ;
\draw [shift={(379,112)}, rotate = 180] [fill={rgb, 255:red, 74; green, 144; blue, 226 }  ,fill opacity=1 ][line width=0.08]  [draw opacity=0] (8.93,-4.29) -- (0,0) -- (8.93,4.29) -- cycle    ;
\draw [color={rgb, 255:red, 74; green, 144; blue, 226 }  ,draw opacity=1 ]   (337,164) -- (373.05,157.53) ;
\draw [shift={(376,157)}, rotate = 169.82] [fill={rgb, 255:red, 74; green, 144; blue, 226 }  ,fill opacity=1 ][line width=0.08]  [draw opacity=0] (8.93,-4.29) -- (0,0) -- (8.93,4.29) -- cycle    ;
\draw [color={rgb, 255:red, 74; green, 144; blue, 226 }  ,draw opacity=1 ]   (337,217) -- (379.19,161.39) ;
\draw [shift={(381,159)}, rotate = 127.18] [fill={rgb, 255:red, 74; green, 144; blue, 226 }  ,fill opacity=1 ][line width=0.08]  [draw opacity=0] (8.93,-4.29) -- (0,0) -- (8.93,4.29) -- cycle    ;
\draw [color={rgb, 255:red, 74; green, 144; blue, 226 }  ,draw opacity=1 ]   (337,217) -- (376.05,224.44) ;
\draw [shift={(379,225)}, rotate = 190.78] [fill={rgb, 255:red, 74; green, 144; blue, 226 }  ,fill opacity=1 ][line width=0.08]  [draw opacity=0] (8.93,-4.29) -- (0,0) -- (8.93,4.29) -- cycle    ;
\draw [color={rgb, 255:red, 74; green, 144; blue, 226 }  ,draw opacity=1 ]   (337,270) -- (377,270) ;
\draw [shift={(380,270)}, rotate = 180] [fill={rgb, 255:red, 74; green, 144; blue, 226 }  ,fill opacity=1 ][line width=0.08]  [draw opacity=0] (8.93,-4.29) -- (0,0) -- (8.93,4.29) -- cycle    ;
\draw [color={rgb, 255:red, 74; green, 144; blue, 226 }  ,draw opacity=1 ]   (337,323) -- (377,323) ;
\draw [shift={(380,323)}, rotate = 180] [fill={rgb, 255:red, 74; green, 144; blue, 226 }  ,fill opacity=1 ][line width=0.08]  [draw opacity=0] (8.93,-4.29) -- (0,0) -- (8.93,4.29) -- cycle    ;

\draw (42,48) node [anchor=north west][inner sep=0.75pt]   [align=left] {\textbf{FR no-go theorem}};
\draw (314,48) node [anchor=north west][inner sep=0.75pt]   [align=left] {\textbf{GHZ--FR agreement (truth) no-go theorem}};
\draw (199,103) node [anchor=north west][inner sep=0.75pt]  [color={rgb, 255:red, 255; green, 255; blue, 255 }  ,opacity=1 ] [align=left] {(U)};
\draw (414,264) node [anchor=north west][inner sep=0.75pt]  [color={rgb, 255:red, 255; green, 255; blue, 255 }  ,opacity=1 ] [align=left] {Classical agreement (AOE)};
\draw (388,104) node [anchor=north west][inner sep=0.75pt]  [color={rgb, 255:red, 255; green, 255; blue, 255 }  ,opacity=1 ] [align=left] {Universality of Quantum Theory};
\draw (222,230) node [anchor=north west][inner sep=0.75pt]  [font=\Large,rotate=-180] [align=left] {$\displaystyle < $};
\draw (222,336) node [anchor=north west][inner sep=0.75pt]  [font=\Large,rotate=-180] [align=left] {$\displaystyle < $};
\draw (223,283) node [anchor=north west][inner sep=0.75pt]  [font=\Large,rotate=-180] [align=left] {$\displaystyle < $};
\draw (198,155) node [anchor=north west][inner sep=0.75pt]  [color={rgb, 255:red, 255; green, 255; blue, 255 }  ,opacity=1 ] [align=left] {(S)};
\draw (100,208) node [anchor=north west][inner sep=0.75pt]  [color={rgb, 255:red, 255; green, 255; blue, 255 }  ,opacity=1 ] [align=left] {(Q)};
\draw (100,261) node [anchor=north west][inner sep=0.75pt]  [color={rgb, 255:red, 255; green, 255; blue, 255 }  ,opacity=1 ] [align=left] {(C)};
\draw (100,314) node [anchor=north west][inner sep=0.75pt]  [color={rgb, 255:red, 255; green, 255; blue, 255 }  ,opacity=1 ] [align=left] {(D)};
\draw (298,208) node [anchor=north west][inner sep=0.75pt]  [color={rgb, 255:red, 255; green, 255; blue, 255 }  ,opacity=1 ] [align=left] {(Q')};
\draw (298,261) node [anchor=north west][inner sep=0.75pt]  [color={rgb, 255:red, 255; green, 255; blue, 255 }  ,opacity=1 ] [align=left] {(C')};
\draw (298,314) node [anchor=north west][inner sep=0.75pt]  [color={rgb, 255:red, 255; green, 255; blue, 255 }  ,opacity=1 ] [align=left] {(D')};
\draw (410,313) node [anchor=north west][inner sep=0.75pt]  [color={rgb, 255:red, 255; green, 255; blue, 255 }  ,opacity=1 ] [align=left] {\textit{only classical agents reason}};
\draw (412,150) node [anchor=north west][inner sep=0.75pt]  [color={rgb, 255:red, 255; green, 255; blue, 255 }  ,opacity=1 ] [align=left] {Personal knowledge (AOE)};
\draw (433,220) node [anchor=north west][inner sep=0.75pt]  [color={rgb, 255:red, 255; green, 255; blue, 255 }  ,opacity=1 ] [align=left] {Born Compatibility };

\end{tikzpicture}  \caption{Comparing the assumptions of the FR paradox to those of the GHZ--FR no-go theorems. An assumption of the FR no-go theorem being strictly stronger than a corresponding assumption of the GHZ--FR agreement no-go theorem is denoted by $>$. In red we compare the assumptions (U), (S), (Q), (C) and (D) of the FR no-go theorem to the rephrased GHZ--FR assumptions (U), (S), (C'), (D') and (Q'). How these rephrased assumptions correspond to the GHZ--FR assumptions of \Cref{sec:GHZ_FR_nogo}, which we have used throughout this work, is shown by the arrows between the assumptions in red circles and blue boxes. The assumptions of the agreement no-go theorem are listed, as the agreement no-go theorem is stronger than the truth one. In the truth no-go theorem, {\protect\PersK} and {\protect\ClassicalT} are replaced by the stronger {\protect\AOEfull}, as listed in brackets.} \label{fig:assumptions_FR_GHZ_FR} 
\end{figure}
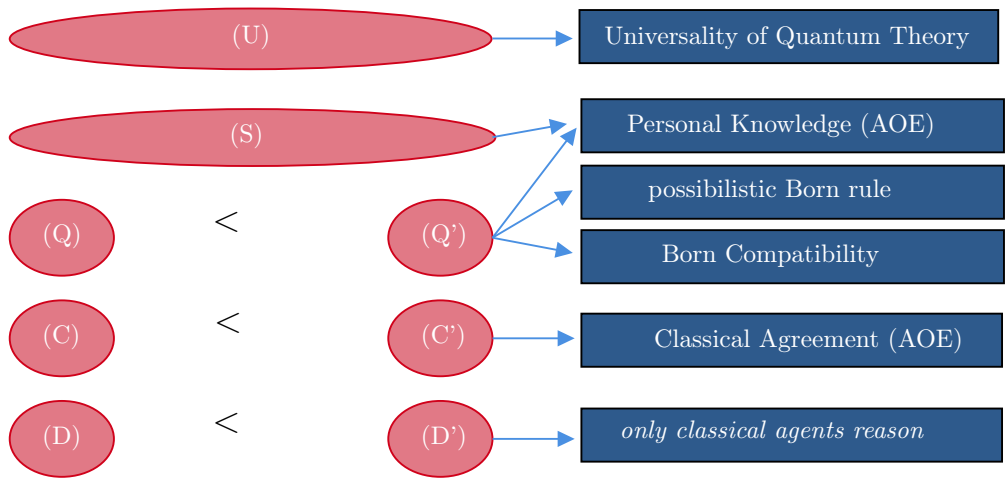

\subsection{Related work on the FR paradox} \label{sec:previous_work}
In this section, we will compare the GHZ--FR paradox to previous work on the FR paradox and protocols related to the GHZ--FR paradox.

\paragraph*{Previous responses to the FR paradox}
The FR no-go theorem \cite{frauchiger2018quantum,vilasini2022general} forces us to refute one of its assumptions, and has been discussed extensively in the literature. 
We are not able to survey all relevant works; instead, we aim to provide a concise overview.
We classify previous responses to the FR paradox into three main categories.\footnote{Some works belong to more than one category, so we mention them more than once.} For more extended reviews we point the reader to Refs.~\cite{vilasini2022general,nurgalieva2020testing}.

\medskip

First, there is work that questions and re-analyses the assumptions used in the FR paradox \cite{vilasini2022general,nurgalieva2018inadequacy,sudbery2019hidden,nurgalieva2020testing,montanhano2023contextuality,boge2019quantum,vilasini2019multi,fraser2020fitch,schmid2023review}.
Some implicit assumptions have been identified in Refs.~\cite{nurgalieva2018inadequacy,sudbery2019hidden}. 
A link between paradoxes in FR-like scenarios and multi-agent knowledge structures, epistemic logic \cite{fagin2004reasoning,blackburn2001modal}, and contextuality has been investigated in Refs.~\cite{montanhano2023contextuality,boge2019quantum,nurgalieva2018inadequacy,vilasini2019multi,fraser2020fitch,corti2023logico,baltag2023logic}.

Other works claim that the FR paradox contains implicit assumptions that can be deemed invalid a priori, thereby questioning the relevance or novelty of the FR paradox and no-go theorem \cite{aaronson2018s,healey2018quantum,araujo2018flaw,fortin2019wigner,drezet2018wigner}. 
Aaronson \cite{aaronson2018s} argued in a blog post that the reasoning statements made by the friends are invalid at the end of the protocol is that the brains of the friends have been `Hadamarded' in the process. 
It is not clear to us whether this argument claims that (i) only reasoning or (ii) also mere experiencing of the measurement result by the unitarily-modelled quantum observers is invalid. 
Case (i) has been criticised by del Rio and Renner~\cite{del2024reply}, and anyway does not apply to the GHZ--FR paradox, as only `classical' agents reason.
Case (ii) could correspond to rejecting {\universality}, but may also suggest a relativity of observed events, depending on its exact interpretation.
Healey \cite{healey2018quantum} argues that an additional assumption of Intervention Sensitivity is required for reasoning about outcomes obtained by quantum agents,
and Drezet \cite{drezet2018wigner} asserts that in a Bohrian interpretation reasoning by quantum agents may be problematic.
Vilasini and Woods \cite{vilasini2022general} suggest that the arguments of Refs.~\cite{healey2018quantum,araujo2018flaw} also refer to the post-selection step, but this depends on the exact reading of these two references.
Criticisms of post-selection and reasoning by quantum agents do not apply to the GHZ--FR paradox as no post-selection or reasoning by quantum agents is needed.  
Ara\'ujo \cite{araujo2018flaw} seems to argue that both collapse and no-collapse modelling of a measurement is needed for the FR argument. However, we note that the state-update rule is not needed \cite{fortin2019wigner}, and only unitary modelling and the Born rule for predictions are used. Namely, each use the Born rule in the FR involves only measurements that can be performed jointly.\footnote{We thank David Schmid and Yìlè Y\={\i}ng for discussions on this matter.} 
Fortin and Lombardi \cite{fortin2019wigner} argue that the FR argument is either illegitimate or uninteresting and simply another proof of contextuality as agents make reasoning statements that correspond to different contexts when mapping the FR protocol to a Bell nonlocality scenario, but does not respond to how agents can then use quantum theory consistently in the presence of superobservers.

\medskip

Secondly, the FR paradox has been analysed from the viewpoint of specific interpretations of quantum theory \cite{nurgalieva2020testing,lazarovici2019quantum,drezet2018wigner,sudbery2017single,sudbery2019hidden,di2021stable,cavalcanti2023consistency,gambini2019single,debrota2020respecting,losada2019frauchiger}, and arguments have been {put forward for which assumption of the FR paradox to refute.
The FR paradox has been studied in Bohmian and Bell--Bohmian theories in Refs.~\cite{lazarovici2019quantum,drezet2018wigner,sudbery2017single}, in a relational quantum theory setting in Ref.~\cite{di2021stable}, from a QBist perspective in Ref.~\cite{debrota2020respecting}, and in the consistent-histories interpretation in Ref.~\cite{losada2019frauchiger}.
An overview of the FR paradox in different interpretations of quantum mechanics can be found in Ref.~\cite{nurgalieva2020testing}. 
\begin{itemize}
\item The (Bell--)Bohmian analyses of the FR \cite{lazarovici2019quantum,drezet2018wigner,sudbery2017single,sudbery2019hidden} argue that the 
agents apply the Born rule (Q) wrongly, as they do not take into account the full state, \ie the pilot wavefunction.
But still, one would like the agents to be able to make correct predictions, even if they cannot take an outsider's view. 
We have provided a possible way to phrase their predictions correctly in \Cref{sec:ways_to_resolve_trust} under `Reject Born Compatibility'.

\item The relational and QBist perspectives propose to refute FR's consistency assumption (C),
allowing for different agents to assign different truth values to statements \cite{nurgalieva2020testing,di2021stable,debrota2020respecting}. One may argue that also the consistent-histories approach to the FR \cite{losada2019frauchiger} refutes (C) \cite{nurgalieva2020testing}, as inferences are made between inconsistent sets of histories.
When refuting (C), one would still like to know when one agent can trust another agent's statement or reasoning. 
For example, in the GHZ--FR paradox, one can imagine that Ursula's, Valentina's, and Wigner's measurements are all performed by the same agent (sequentially in their proper time), who then makes statements similar to those of Ursula, Valentina, and Wigner: must this agent then not trust their own statement? 
In the QBist approach to FR by DeBrota, Fuchs, and Schack \cite{debrota2020respecting}, it is for example argued that when to trust another agent's statement can be decided on a case-by-case basis.
Our GHZ--FR no-go theorem shows that even classical communication cannot be trusted when allowing different truth values for different agents, and we have argued that perhaps it is easier to instead restrict the validity of statements using the Born rule for predictions, while holding on to consistency between agents. 
The fact that no quantum agents need to reason in the GHZ--FR paradox is crucial for this argument, and poses sharper restrictions on building consistent reasoning principles.\footnote{Whether the relative-facts approach of Ref.~\cite{di2021stable}, which refutes (C) and argues to not believe the reasoning of the quantum agents, resolves the paradoxes presented in this paper or not depends on its exact interpretation, as discussed in \Cref{appendix:stable_relative_GHZ_FR}.} 

\item Other works \cite{relano2018decoherence,relano2020decoherence,kastner2020unitary,zukowski2021physics,gambini2019single} argue for refuting the existence of superobservers because of decoherence or absolute collapse, rendering the FR protocol -- involving superobservers -- impossible, even in principle. 
However, decoherence may be part of an explanation of the quantum-to-classical transition, but it is not a full answer to the measurement problem \cite{schlosshauer2019quantum,schlosshauer2005decoherence,schlosshauer2007quantum}. Absolute collapse requires an explanation of how or why the collapse happens, as in spontaneous collapse theories \cite{ghirardi1985model,ghirardi1986unified} or gravity-induced collapse \cite{penrose1996gravity,diosi1987universal}. 
Experiments may shed light on whether collapse theories, which have different dynamics than ordinary quantum theory, are likely to be correct \cite{bassi2003dynamical,forstner2020nanomechanical,arnquist2022search,donadi2021novel,carlesso2022present}.
\end{itemize}

\medskip

Thirdly, there is some work towards responding to Renner's challenge of building a consistent set of reasoning principles for quantum theory
that evades contradictions like the FR paradox \cite{renes2021consistency,narasimhachar2020agents,polychronakos2022quantum,vilasini2022general}. 
Here we note that Vilasini and Woods~\cite{vilasini2022general} propose that one be explicit about which Heisenberg cuts are used in each statement concerning outcomes of performed measurements. This leaves open the question of the physical meaning of Heisenberg cuts, and why should statements be formulated as dependent on it? Is it merely a choice of what to model inside and outside the quantum theory? If so, then still desirably no contradictions about the outcome values as in the (GHZ--)FR paradox should arise when making different choices, if it is merely a choice. However, a resolution à la {\ROE} proposes a reason of why such mentioning would be necessary: it refers to the statement being conditional on physically obtaining (a copy of) the outcome record(s) one is talking about; a prediction by agent $A$ is meant for when $A$ actually obtains those outcome records (and the protocol be performed truthfully per agent $A$'s knowledge).}
Moreover, we also note that the GHZ--FR no-go theorem invalidates the set of reasoning principles from Ref.~\cite{polychronakos2022quantum} as we show in \Cref{appendix:GHZ_FR_invalidates_reasoning_principles_polychronakos}.

\medskip

Our work belongs in the first category, as it provides a stronger no-go theorem than the FR. It also touches on the third category by suggesting ways to resolve paradoxes involving superobservers in quantum theory.

\paragraph*{Related work on the GHZ--FR paradox}
A protocol similar to the one underlying the GHZ--FR paradox has been proposed by {\.Z}ukowski and Markiewicz \cite{zukowski2021physics} and by Leegwater \cite{leegwater2022greenberger}. While the measurements performed in the protocols (without reasoning steps) in Refs.~\cite{zukowski2021physics,leegwater2022greenberger} are the same,
the exact way in which a contradiction is obtained differs, as do the ensuing analyses and no-go theorems. 
In Ref.~\cite{zukowski2021physics}, an explicit motivation for how the different measurement contexts that lead to a contradiction are chosen and considered together is not provided. 
In Ref.~\cite{leegwater2022greenberger} this motivation comes from simultaneity in special relativity. 
However, in cases where the motivation from special relativity of Ref.~\cite{leegwater2022greenberger} cannot be applied, the GHZ--FR paradox still yields a contradiction, as discussed in \Cref{appendix:related_work}.

Going beyond Ref.~\cite{leegwater2022greenberger}, we investigate the relationship of the GHZ--FR protocol to the FR paradox more thoroughly, and expose the hierarchy of their underlying nonlocality models.
Unlike Refs.~\cite{zukowski2021physics,leegwater2022greenberger} our agreement no-go theorem requires no absolute outcomes, due to the reasoning part of our protocol being different. 
This important weakening allowed us to propose the {\ROE} resolution as replacing {\AOEfull} in \Cref{sec:GHZ_nogo_trust,sec:possible_resolution}.
Unlike these works \cite{zukowski2021physics,leegwater2022greenberger}, we move away from having a bird's-eye view on the protocol as an outsider, and instead consider the viewpoints from observers inside the protocol who try to make predictions and reason, as was the goal in the original FR paradox. To relate to the notion of \textit{Born Inaccessible correlations}\footnote{Born Inaccessible Correlations is defined in \cite{schmid2023review} as: ``It requires that when two measurements are made in parallel, the two outcomes always arise with a joint frequency given by the Born rule, even if no single observer could access both outcomes, even in principle.''} of Ref.~\cite{schmid2023review}, we acknowledge that there will always be correlations inaccessible to any agent. However, we assign the label `inaccessible' as we can take the outsider's bird's-eye view on the protocol; it might be impossible for agents inside to have this knowledge in all situations. Were we to make a no-go theorem using the notion of Born Inaccessible correlations, we would find a no-go theorem similar to the GHZ--FR truth no-go theorem, replacing Born Compatibility with Born Inaccessible correlations.
We discuss the works~\cite{leegwater2022greenberger,zukowski2021physics} further in \Cref{appendix:related_work}. 

A similar relational resolution as {\ROE} in the context of Relational Quantum Mechanics \cite{rovelli1996relational,rovelli2018space}
was discussed by Cavalcanti, Di Bagio and Rovelli \cite{cavalcanti2023consistency} for a scenario involving the same measurement protocol as the GHZ--FR paradox, based on Refs.~\cite{zukowski2021physics,lawrence2023relative}. 
Our resolution per {\ROE} accords with their proposed resolution, except that our proposal does not invoke the existence of cross-perspective links \cite{adlam2022information}.
Furthermore, as noted by Cavalcanti and Wiseman in Ref.~\cite{cavalcanti2021implications}, to resolve the Local Friendliness no-go theorem, by giving up on {\AOE} for a notion of relative events one can maintain compatibility with Leibniz's principle.

\paragraph*{Relation to the Local Friendliness no-go theorem} 
The Local Friendliness no-go theorem \cite{bong2020strong,haddara2022possibilistic,cavalcanti2021implications,schmid2023review} is based on a protocol similar to the FR's but is phrased in a device-independent fashion. 
It uses the Born rule only for outcomes that are obtained by us (\ie that are available after the protocol),
but in addition to {\AOE} it also requires a notion of free choices satisfying `Local Agency', a locality constraint not invoking hidden variables as in Bell's local causality.
If in the GHZ--FR paradox of \Cref{sec:GHZ_final}
we give the superobservers a choice between performing their measurement or asking their respective observers for their outcomes, we obtain a Local Friendliness scenario \cite{bong2020strong,haddara2022possibilistic}. 
In upcoming work we\footnote{Joint work of LW with R. Wagner, Y. Y{\=\i}ng, and D. Schmid.} clarify how also contextuality scenarios may be used to produce a Local Friendliness no-go theorem.

\section{Conclusion} \label{sec:conclusion}
We conclude by summing up the main results and outline some future research directions.

\subsection{Summary}
Motivated by the question of what the FR paradox \cite{frauchiger2018quantum} implies for quantum foundations, we presented a similar but strictly stronger paradox, the GHZ--FR paradox, based on the strongly contextual GHZ--Mermin model \cite{greenberger1989going,mermin1990simple}.
Outlining carefully the required assumptions, we obtained two no-go theorems of increasing strength, which we named the GHZ--FR truth and agreement no-go theorems.
The easiest to grasp is the GHZ--FR truth no-go theorem, whose assumptions are summarised in \Cref{table:GHZ_FR_truth_assumptions}. 
The agreement no-go theorem weakens {\AOEfull} replacing it by {\PersK} and {\ClassicalT}; see \Cref{table:AOE_trust_truth}: instead of \emph{absolute} outcome assignments to performed measurements, it posits outcome assignments that may a priori differ per agent, subject to the condition that classically communicating agents must agree on their outcome assignments.

\begin{table}[h]
\resizebox{\textwidth}{!}{%
\begin{tabular}{|
>{\columncolor[HTML]{F0EEEC}}l 
>{\columncolor[HTML]{F0EEEC}}l |}
\hline
\multicolumn{2}{|l|}{\cellcolor[HTML]{656565}{\color[HTML]{FFFFFF} }}                                                       \\
\multicolumn{2}{|l|}{\multirow{-2}{*}{\cellcolor[HTML]{656565}{\color[HTML]{FFFFFF} \textbf{GHZ--FR truth no-go theorem}}}} \\ \hline
                                           &                                                                                                 \\

Absoluteness of Observed Events &
  \begin{tabular}[c]{@{}l@{}}Every performed measurement has an absolute, single-valued outcome. \end{tabular}\\ (AOE) &
          \\
Born Compatibility &
  \begin{tabular}[c]{@{}l@{}} Outcomes assigned (per \text{AOE}) to a set of measurements that an agent $A$ \\ believes she can gather jointly in a memory  must be compatible with $A$'s use \\ of the possibilistic Born rule.
 \end{tabular} \\  &
                           \\
Universality of Quantum Theory             & Superobservers exist (in principle) and quantum theory is correct.             \\
                                                       
  \hline
\end{tabular}%
}
\caption{Assumptions of the GHZ--FR truth no-go theorem; more precisely stated assumptions can be found in \Cref{sec:GHZ_nogo}.}
 \label{table:GHZ_FR_truth_assumptions}

\end{table}

\begin{table}[h]
\resizebox{\textwidth}{!}{%
\begin{tabular}{|
>{\columncolor[HTML]{F0EEEC}}l 
>{\columncolor[HTML]{F0EEEC}}l 
>{\columncolor[HTML]{F0EEEC}}l |}
\hline
\cellcolor[HTML]{656565}{\color[HTML]{FFFFFF} }                                                       & \multicolumn{2}{l|}{\cellcolor[HTML]{656565}{\color[HTML]{FFFFFF} }}                                                                                                                                             \\
\multirow{-2}{*}{\cellcolor[HTML]{656565}{\color[HTML]{FFFFFF} \textbf{GHZ--FR truth no-go theorem}}} & \multicolumn{2}{l|}{\multirow{-2}{*}{\cellcolor[HTML]{656565}{\color[HTML]{FFFFFF} \textbf{GHZ--FR agreement no-go theorem}}}}                                                                                       \\ \hline
\multicolumn{1}{|l|}{\cellcolor[HTML]{F0EEEC}}                                                        & Personal Knowledge & \begin{tabular}[c]{@{}l@{}}Every performed measurement can be assigned a \\ single-valued outcome by an agent who believes \\ they are able to gather that outcome.\end{tabular} \\
\multicolumn{1}{|l|}{\multirow{-2}{*}{\cellcolor[HTML]{F0EEEC}Absoluteness of Observed Events}}       &                    &                                                                                                                                                                                             \\
\multicolumn{1}{|l|}{\cellcolor[HTML]{F0EEEC}}                                                        & Classical Agreement    & \begin{tabular}[c]{@{}l@{}}The outcome assignments of classical agents 
 who \\ communicate classically must agree on their overlap. \end{tabular}                                                      \\ \hline
\end{tabular} } 
\caption{The GHZ--FR agreement no-go theorem is based on the same assumptions as the GHZ--FR truth no-go theorem, but with Personal Knowledge and Classical Agreement replacing Absoluteness of Observed Events.}
\label{table:AOE_trust_truth}
\end{table}

Pursuing the route of rejecting {\AOE} or {\BornCAOE}, as discussed in \Cref{sec:GHZ_nogo_trust}, the agreement no-go theorem suggests a principle that we call {\BornP}, defined in \Cref{sec:possible_resolution}, meaning that an agent $A$ can use the Born rule only to formulate predictions that condition on $A$ actually learning about all involved outcomes.
This principle can be seen as a natural agent-based extension of Peres's dictum \cite{peres1978unperformed} `Unperformed experiments have no results':
\begin{quote}
   `Unperformed experiments have no results, and unknown results have no values'.
\end{quote}

One may view this as a natural quantum-vs-classical statement.
Therefore, as a possible resolution to extended Wigner's friend paradoxes we propose {\BornP} as a principle governing agents' use of quantum theory in a world where {\universality} holds.
In a relational framework, rejecting {\AOE}, {\BornP} can be implemented by having single-valued outcomes only defined relative to the experimenter and other agents who learn about the outcome, which we named {\ROE} in \Cref{sec:GHZ_nogo_trust}. 
In a hidden-variable framework, rejecting {\BornCAOE}, one may instead have hidden variables determining absolute, single-valued outcomes of all performed measurements but agents must follow {\BornP} for correct predictions of the Born rule.
Of course, one may also question the existence of superobservers, or simply allow for classically communicating agents to have inconsistent truth assignments. The latter, however, seems impractical; see \Cref{sec:ways_to_resolve_trust,sec:previous_work}.

\subsection{Outlook}
To conclude, our work opens up various directions for future work.
In upcoming work, we will discuss the implications of the the GHZ--FR paradox for the measurement problem, for different interpretations of quantum theory, and how it suggests a resolution to Renner's challenge of providing a consistent set of reasoning rules in quantum settings.
Another intriguing question is what implications the GHZ--FR paradox and other extended Wigner's friend paradoxes have for the discussion on the epistemic or ontic nature of quantum states \cite{harrigan2010einstein,pusey2012reality,hardy2013quantum,leifer2014quantum} or ontologies of quantum mechanics.
A deeper investigation of the relationship between {\ROE} or {\BornP} and other proposed resolutions of extended Wigner's friend paradoxes,
including bubbles and quantum influences~\cite{ormrod2024quantum,cavalcanti2021view}, the relational approach of `fact-nets'~\cite{martin2023fact}, and quantum reference frames~\cite{de2020quantum} is left for future work.
Finally, in recent work extended Wigner's friend paradoxes have been presented based on contextuality rather than nonlocality \cite{walleghem2023extended,szangolies2020quantum,walleghem2024connecting}. 
An interesting avenue to explore is to combine
Wigner's friend with other nonclassical quantum features such as steering~\cite{schrodinger1935discussion,wiseman2007steering}, incompatibility~\cite{guhne2023colloquium}, generalised contextuality~\cite{spekkens2005contextuality,schmid2020structure}, coherence \cite{streltsov2017colloquium,galvao2020quantum,Wagner_2024}, \ldots

\section*{Acknowledgements}
L.W.~expresses deep gratitude to E. G. Cavalcanti for many stimulating and insightful discussions on the Local Friendliness no-go theorem, the FR paradox, and the concept of `relative events'. L.W.~also thanks INL in Braga, Portugal for the kind hospitality, and its QLOC group at INL for many useful comments and interesting discussions on contextuality. L.W. also thanks Rafael Wagner, Nick Ormrod, David Schmid, and Y{\`\i}l{\`e} Y{\=\i}ng for many interesting discussions. 
R.S.B.~thanks Shane Mansfield for discussions several years ago which planted the seed of some of the present results.
Figures are produced with tikz and matcha, with the help of Rafael Wagner.

L.W.~acknowledges support from the United Kingdom Engineering and Physical Sciences Research Council (EPSRC) through the DTP Studentship EP/W524657/1. 
R.S.B.~acknowledges financial support from FCT -- Funda\c{c}{\~a}o para a Ci{\^e}ncia e a Tecnologia (Portugal) through CEECINST/00062/2018.

\clearpage

\appendix

\section{Clarification of concepts}
\label{appendix:clarific}

In this section we clarify the concept of `learning about an outcome'.

\paragraph*{Definition of `learning about an outcome'}
    By \emph{learning about an outcome of a measurement}, or \emph{asking an experimenter for the outcome of their performed measurement}, we mean the following. Consider an agent $A$ performing a measurement on a qubit system $S$ in the basis $\ket{0}, \ket{1}$. Let $E$ be the environment, namely the whole closed system excluding $A$ and $S$. We can model this measurement as a unitary: \begin{equation}
        \big( \alpha \ket{0} + \beta \ket{1} \big)_S \otimes |\text{ready} \rangle_A \otimes | \text{ready} \rangle_E \rightarrow \alpha \ket{0}_S \otimes \ket{0}_A \otimes \ket{0}_E + \beta \ket{1}_S \otimes \ket{1}_A \otimes \ket{1}_E
    \end{equation} where we assume that the agent and environment are initialised in the $|\text{ready}\rangle_A, |\text{ready} \rangle_E$. Here $\ket{0}_E, \ket{1}_E$ denote the traces of the measurement outcome $0$ or $1$ in the environment. Learning about the outcome of this measurement then includes: \begin{itemize}
        \item Asking the experimenter $A$ for the outcome, corresponding to measuring $A$ in the basis $\ket{0}_A,\ket{1}_A$;
        \item Measuring the system $S$ in the measurement basis $\ket{0}_S,\ket{1}_S$;
        \item Obtaining the measurement outcome through traces in the environment, corresponding to measuring the environment $E$ in the basis $\ket{0}_E,\ket{1}_E$;
        \item classical communication about outcomes through another agent who learns about the outcome in one of the three ways stated above.
    \end{itemize} 

\section{Proof of the GHZ--FR agreement no-go theorem} \label{appendix:proof_agreement_nogo}
\begin{theorem}[GHZ--FR agreement no-go theorem]
   The GHZ--FR paradox shows that the assumptions  {\universality}, {\possBorn}, {\PersK}, {\ClassicalT}, and {\BornCPersK}  are incompatible.
\end{theorem}
\begin{proof} 
    The assumption {\universality} ensures that the GHZ--FR measurement protocol can be performed. 
    We consider the reasoning of the superobservers after having obtained their outcomes.
    Ursula makes a statement for $U,B,C$, believing she might test her prediction still.
    Using {\PersK}, we denote Ursula's outcome assignment to the performed measurements $U,B,C$ by $f_U:\{U,B,C\}\rightarrow O$.
    Similarly, we denote Valentina's and Wigner's outcome assignments by $f_V:\{V,A,C\} \rightarrow O$ and $f_W:\{W,A,B\} \rightarrow O$, respectively.
    First, we consider Ursula's reasoning statement. For her, assuming {\BornCPersK} her personal outcome assignments must be compatible with her use of the {\possBorn}:
    \begin{equation} \label{eq:ubc_proof_trust}
        f_U(B) \oplus f_U(C) = 1 \oplus f_U(U).
    \end{equation}
    Similarly, we obtain for Valentina and Wigner that 
    \begin{align}
    \label{eq:avc_proof_trust}
        f_V(A) \oplus f_V(C) &= 1 \oplus f_V(V), \\
        \label{eq:abw_proof_trust}
        f_W(A) \oplus f_W(B) &= 1 \oplus f_W(W).
    \end{align}
   As the outcomes of Ursula, Valentina, and Wigner are all classically available after the protocol, and can be obtained by us for example too, we have an outcome assignment $f_Z: \{U,V,W\} \rightarrow O$, satisfying
    \begin{equation} \label{eq:uvw_proof_trust_Z}
        f_Z(U) \oplus f_Z(V) \oplus f_Z(W) = 0,
    \end{equation}
    by application of {\possBorn}. As we can communicate with Ursula about her outcome, we have to assign the same value to that outcome, \ie $f_Z(U)=f_U(U)$, and similarly we will have $f_Z(V)=f_V(V), f_Z(W)=f_W(W)$, so we denote them by $u,v,w$. This is a verifiable use of quantum theory, and can also be seen as a verifiable application of \PersK, \BornCPersK and \ClassicalT. 
    Finally, as the agents Ursula, Valentina and Wigner classically communicate, their outcome assignments must agree on overlaps of their assignments by \ClassicalT, \ie
    \begin{equation}
    f_V(A) = f_W(A),  \qquad
     f_U(B) = f_W(B), \qquad 
    f_U(C) = f_V(C),
    \end{equation}
    so we can denote them simply by $a,b,c$, respectively.\footnote{Note that for example $f_U(C)=f_V(C)$, \ie Ursula and Valentina assigning the same outcome value to Charlie's measurement, could be empirically verified by having them ask Charlie for his outcome if Wigner would perform his measurement later. A similar observation holds for $f_U(B) = f_W(B)$ and $f_U(C)=f_V(C)$.} 
    Combining with the conditions derived above yields
     \begin{equation} \label{eq:abc_predictions_GHZ_FR}
        \begin{split}
            b \oplus c = 1 \oplus u, \quad
            a \oplus c = 1 \oplus v, \quad
            a \oplus b = 1 \oplus w, \quad
            u \oplus v \oplus w = 0.
        \end{split}
\end{equation} 
A contradiction is obtained as these four equations are inconsistent,
\ie they admit no solution over outcomes in $\{0,1\}$.
\end{proof}

\section{The GHZ--FR paradox in epistemic logic} \label{appendix:epistemic_logic}
In \Cref{remark:epistemic} we have shortly commented on the GHZ--FR paradox in terms of epistemic logic. We provide a more detailed discussion here. 
    In Refs.~\cite{vilasini2019multi,nurgalieva2018inadequacy,montanhano2023contextuality}, the FR paradox has been described in terms of epistemic modal logic \cite{fagin2004reasoning,blackburn2001modal}.
    We briefly outline a similarly flavoured treatment of the GHZ--FR paradox.

    Epistemic logic is suitable to reason about knowledge or beliefs (of different agents).
    Syntactically, besides the usual logical connectives, formulas can be built using a unary modal operator $K_A$ for each agent $A$,
    where a formula $K_A \phi$ is interpreted as meaning `agent $A$ believes that $\phi$ is True'. 
    Different sets of axioms are typically considered in epistemic logic systems, varying with the intended interpretation.
    
    In Ref.~\cite{nurgalieva2018inadequacy}, the \textit{Distribution Axiom} of modal logic (a.k.a.\ axiom \textbf{K}) is listed as assumption (D) of the FR paradox
    (see \Cref{sec:comparison_FR}), and
    used to justify the use of \textit{modus ponens} by quantum agents in the FR paradox. 
    In the GHZ--FR paradox, however, as in the version of the FR paradox from \Cref{sec:stricter_FR}, only classical agents need to reason, 
    dispensing with the reasoning by quantum agents considered in the earlier treatments of the FR paradox \cite{nurgalieva2018inadequacy,vilasini2019multi,montanhano2023contextuality,frauchiger2018quantum}.
    The fact that we need only consider reasoning by classical agents in the protocol is the reason why we do not state such reasoning as a separate assumption in our no-go theorems.
    In modal logical terms, this is reflected in the fact that we do not need the epistemic modalities $K_A$ for quantum agents:
    the whole argument can be phrased in a logical system with modalities only for each \textit{classical} agent in the protocol. 
    It could perhaps be argued that this somewhat weakens the appeal of a formalisation of the paradoxes using modal logic.

    Still, the core assumptions of our GHZ--FR no-go theorems can be phrased in these terms.
    The GHZ--FR truth no-go theorem requires the existence of an absolute truth,
    per {\AOEfull}.
    This is captured by the so-called \textit{Truth} or \textit{Knowledge Axiom} (a.k.a.\ axiom \textbf{T}),
    which typically distinguishes knowledge from mere belief.
    It states that if an agent \textit{knows} something, then it is true:
    \begin{equation}
        K_A \phi \rightarrow \phi.
    \end{equation}
    In the GHZ--FR agreement no-go theorem, such absolute truth is replaced by {\PersK} and {\ClassicalT}. 
    {\PersK} is captured simply by the fact that one considers modal operators $K_A$ for each classical agent $A$, together with the restriction that to form $K_A\phi$ the formula $\phi$ can only use propositions about outcomes of measurements accessible to $A$.\footnote{The semantic interpretation is precisely that there exists a valuation $\alpha_A$ assigning values in $0,1$ to the propositional variables in the formula $\phi$.}
    {\ClassicalT} corresponds to agreement between classically communicating agents: if $A$ and $B$ can communicate classically, then
\begin{equation} K_A \phi \rightarrow K_B\phi.\end{equation}
Note that only classical communication between classical agents must be trustworthy, \ie only classical agents must trust each others' classically communicated statements. This assumption is strictly weaker than the stronger epistemic trust structures that were identified for the FR paradox in Refs.~\cite{nurgalieva2018inadequacy,vilasini2019multi,montanhano2023contextuality}, where classical agents must also believe the thoughts made by quantum agents.
Concretely, in the GHZ--FR paradox, the three superobservers Ursula, Valentina, and Wigner can classically communicate and thus must `trust' each other. Hence, $K_U\phi$ if and only if  $K_V\phi$ if and only if  $K_W\phi$.

    The GHZ--FR paradox starts from the following statements, made by each of the superobservers, per  {\possBorn} and {\BornCPersK}:
    \begin{equation} \label{eq:epistemic_GHZ_UVW}
    \begin{split}
        K_U (b\oplus c \oplus u = 1), \\
        K_V (a \oplus c \oplus v = 1), \\
        K_W (a \oplus b \oplus w = 1). 
    \end{split}
    \end{equation}
     Furthermore, any of the classical agents can use the Born rule to make statements about the joint outcomes of the three: in particular, one has
     \begin{equation} K_S (u \oplus v \oplus w = 0)\end{equation}
     for any $S \in \{U, V, W\}$.\footnote{In fact, the assumptions could be stated still somewhat more weakly, as classical agents only need to agree on the overlap of their outcome assignments per \ClassicalT, but for simplicity we stick to the above.}

     Using the Knowledge Axiom ({\AOE}) to remove the modalities, one reaches a contradiction, corresponding to the truth no-go theorem.
     Alternatively, one may use {\ClassicalT} to transfer all the statements under, say the $K_U$ modality, reaching a contradiction for Ursula.
     Finally, we further require the assumption that such a contradiction cannot occur in a modality, \ie that a classical agent cannot believe a contradiction.
     This is expressed by the axiom $K_U \bot \rightarrow \bot$, which is a weaker form of the Knowledge Axiom.

\section{Resolutions to FR-like paradoxes}
\subsection{Resolution of the GHZ--FR paradox} \label{appendix:ROE_resolution}
In the the GHZ--FR paradox, using {\BornP} the three statements by Ursula, Valentina and Wigner are now: 
\begin{itemize}
    \item `If I (Ursula) ask for the outcome of Bob and Charlie, I will find that $b \oplus c = 1$ (as long as none of our memories about our outcomes are altered).'\footnote{The last part `as long as none of our memories about our outcomes are altered' is necessary because of the following. Imagine an agent Zeno were a superobserver for Ursula, Valentina and Wigner, and could thus undo their measurements to go and ask for $a,b,c$, which cannot satisfy the equations for $a,b,c$ found by Ursula, Valentina and Wigner. The precise statements however do not claim to make a prediction in that case, as Ursula's, Valentina's and Wigner's memories would be wiped by Zeno's undoing of their measurements.
}
    \item `If I (Valentina) ask for the outcome of Alice and Charlie, I will find that $a \oplus c = 1$ (as long as none of our memories about our outcomes are altered).'
    \item `If I (Wigner) ask for the outcome of Alice and Bob, I will find that $a \oplus b = 1$ (as long as none of our memories about our outcomes are altered).'
\end{itemize}
The conjunction of these statements is now not paradoxical anymore, as none of Zeno, Ursula, Valentina, Wigner or any other agent asked Alice, Bob and Charlie for the outcomes $a,b,c$. Furthermore, Zeno, obtaining these statements, cannot go and ask for Alice's, Bob's and Charlie's outcomes $a,b,c$. 

 \subsection{Resolution of the FR paradox using the Relativity of Observed Events} \label{appendix:appendix_resolution_FR}
Let us assume {\BornP} as stated in \Cref{sec:GHZ_nogo_trust}, and show how this resolves the FR paradox of \Cref{sec:FR_overview}. Ursula makes a chain of reasoning, using her outcome to find Bob's outcome, who reasons about Alice's, who reasons about Wigner's. However, using {\BornP}, the statements she would make for these agents are as follows:
    \begin{itemize}
        \item Ursula: `If I (Ursula) ask Bob for his outcome, he would always reply that $b=1$ (as long as none of our memories about our outcomes are altered).'
        \item Ursula: `If Bob obtains $b=1$, and he asks Alice for her outcome, Alice replies that $a=1$ (as long as none of our memories about our outcomes are altered).'
        \item Ursula: `If Alice obtains $a=1$, and she asks Wigner for his outcome, then he replies that $w=\textfail$ (as long as none of our memories about our outcomes are altered).'
    \end{itemize}
These statements are now not contradictory as, for example, as no agent can ask for Alice's or Bob's outcome without wiping Ursula's or Wigner's memories. The strengthened version of the FR paradox of \Cref{sec:stricter_FR} is resolved similarly.

\section{Relation to some proposed resolutions in the literature} \label{appendix:relative_facts_and_polychronakos}
\subsection{Does the resolution `stable facts, relative facts' of Di Biagio \& Rovelli resolve FR-like paradoxes?}
\label{appendix:stable_relative_GHZ_FR}
In \cite{di2021stable} a resolution to the FR paradox proposed in terms of stable versus relative facts, applied to the GHZ--FR and hybrid Bell--FR paradox. As we argue below, depending on the exact interpretation of this stable-versus-relative-facts proposal, the GHZ--FR paradox may or may not be resolved. The authors of \cite{di2021stable} refute the Consistency assumption\footnote{Note that we argued in \Cref{sec:comparison_FR_GHZ_FR} that rejecting the possibility for an agent $A$ to assign a single value to the outcome of a measurement performed by another agent $B$, captured in {\AOEfull} and {\PersK} in the GHZ--FR truth and agreement no-go theorems, which suggests assuming {\ROE} instead, corresponds to rejecting (Q) in the FR paradox.} of the original FR paradox \cite{frauchiger2018quantum}, as they state (direct quote): 
\begin{quote}
    ``In terms of relative facts, assumption \textsc{Consistency} implies: “If $W$, applying quantum theory, can be certain that $L_S = a$ relative to $F$, then $W$ can reason as if $L_S = a$ also relative to $W$.” Now, as we have shown, this holds only if every fact relative to $F$ is stable for $W$, which is not a given and depends on the physics. Therefore Assumption Consistency only holds if $S$ or $F$ decohere with respect to $W$.'' [Here $F$ measures the variable $L_S$ of system $S$, with $W$ being a superobserver for $F$. ]
\end{quote} Recall that the Consistency assumption of Frauchiger and Renner can be stated as: \begin{quote}
    \begin{quote}
    `If an agent $A$ knows that an agent $B$ knows that statement $s_1$ is true from reasoning with a theory that $A$ accepts, then agent $A$ can take $s_1$ to be true.'
\end{quote} 
We use a weaker consistency condition in the agreement no-go theorems, namely involving only classical agents who need to agree with each other when they communicate their statements classically.
\end{quote}

\subsubsection*{Stable-versus-relative facts and the GHZ--FR paradox}
Applying the stable-versus-relative-facts resolution method to the GHZ--FR paradox, for example for $U,V,W$ obtaining $u=0,v=0,w=0$, which has a non-zero probability to occur, we obtain the following: \begin{itemize}
    \item From $u=0$ Ursula concludes that $b \oplus c =1$ relative to $B,C$.
    \item From $v=0$ Valentina concludes that $a \oplus c = 1$ relative to $A,C$.
    \item From $w=0$ Wigner concludes that $a \oplus b = 1$ relative to $A,B$. 
\end{itemize} As Ursula, Valentina and Wigner are classical agents and thus have decohered with respect to each other, we can combine these statements. We obtain \begin{equation}
    b \oplus c = 1, \quad a \oplus c = 1, \quad a \oplus b = 1,
\end{equation} relative to $A,B,C$. But these equations do not have a solution as $a,b,c \in \{0,1\}$. Whether or not this is a true contradiction now is debatable and depends on the exact meaning of the stable-vs-relative-facts proposal of \cite{di2021stable}, \ie on the exact meaning of the statement that these equations or outcomes are `relative to $A,B,C$'. If we can apply the usual classical logic to these statements, we run into a contradiction for the outcomes $a,b,c$ relative to $A,B,C$. On the other hand, if one argues that these outcomes $a,b,c$ are relative to $A,B,C$ and thus one cannot combine these equations in a single logical framework, then one does not run into a contradiction. In the latter case, the statements are part of different Boolean algebras and thus cannot be considered together.
In terms of the {\ROE}, one cannot combine these statements as the application of the Born rule translates into statements about outcomes if an agent would go and ask for these outcomes; as agents cannot all ask for these outcomes without violating the protocol the paradox is resolved. The different physical scenarios where agents ask for different outcomes correspond then to the restriction of statements to different Boolean algebras. 

\subsection{The GHZ--FR paradox invalidates a set of proposed reasoning principles}
\label{appendix:GHZ_FR_invalidates_reasoning_principles_polychronakos}
In \cite{polychronakos2022quantum} a set of reasoning principles was proposed as a resolution for the FR paradox. 
In this section we shortly explain how the GHZ--FR paradox invalidates this set of reasoning principles. The author in \cite{polychronakos2022quantum} states: \begin{quote}
    `I argue that the usual rules of quantum mechanics on measurement outcomes have to be complemented with (or rather, understood to imply) the condition that observers who are themselves going to be subject to measurements in a linear combination of macroscopic states cannot make reliable predictions on the results of experiments performed after such measurements.'
\end{quote} In the GHZ--FR paradox only classical agents need to reason, which is allowed by the above principle. Hence the above proposed principle of \cite{polychronakos2022quantum} does not resolve the GHZ--FR paradox. 

\section{Related work on the GHZ--FR paradox} \label{appendix:related_work}
A protocol similar to the GHZ--FR paradox has been proposed by {\.Z}ukowski and Markiewicz in \cite{zukowski2021physics} and by Leegwater in \cite{leegwater2022greenberger}, as explained briefly in \Cref{sec:previous_work}.

\medskip

The assumptions\footnote{The assumptions are not precisely stated in \cite{zukowski2021physics}, but this is our proposal for the required assumptions in \cite{zukowski2021physics}.} in \cite{zukowski2021physics} consist of {\AOEfull}, {\universality} and the application of the Born rule to different Heisenberg cuts in the scenario. It would thus correspond to the GHZ--FR truth no-go theorem, but the agreement no-go theorem is stronger. The assumptions in \cite{leegwater2022greenberger} are: \begin{itemize}
    \item Single Outcomes,\footnote{The assumption Single Outcomes can be seen as the corresponding to {\AOEfull} in the GHZ--FR paradox.} meaning that every performed measurement can be assigned a single-valued outcome, independent of the perspective or frame in which the outcome is described; and
    \item Relativistic Quantum Mechanics, meaning that any inertial frame can be used to describe the evolution of the wavefunction and that the Born rule is valid in any such inertial frame.
\end{itemize} The existence of superobservers is assumed implicitly in \cite{leegwater2022greenberger}. 

The authors of \cite{zukowski2021physics} argue for rejecting the existence of superobservers, but seem to leave some room still for an investigation of having relative single-valued outcomes instead of absolute ones.

In \cite{leegwater2022greenberger} the link with the FR paradox is not investigated, but only mentioned in the following quote: \begin{quote}
    ``It would take too far to examine the exact differences between the two results. However, we can say that in our opinion a crucial difference lies in the fact that F $\&$ R do not seem to assume unitary quantum mechanics in the way we do.''
\end{quote}
Thus there is a close connection between the two scenarios as we have shown here.

\medskip

The GHZ--FR paradox also allows for more varied lightcone structures than in \cite{leegwater2022greenberger}. For example, in the GHZ--FR paradox we can also let Bob perform his measurement in the future lightcone of Charlie, as shown in \Cref{fig:GHZ_FR_leegwater_BafterC}. In this case, Ursula, for example, cannot apply the Born rule in the Relativistic Quantum Mechanics assumption from \cite{leegwater2022greenberger} to $u,b,c$, as $b$ and $c$ never occur simultaneously in some relativistic inertial frame. However, as Charlie's information is not destroyed when Bob performs a measurement, Ursula can still apply {\possBorn}. Thus even in this case we can still run the GHZ--FR protocol, and arrive at a paradox. This situation suggests that the GHZ--FR paradox perhaps does not have so much to do with simultaneity in different relativistic frames.

\begin{figure}[H]
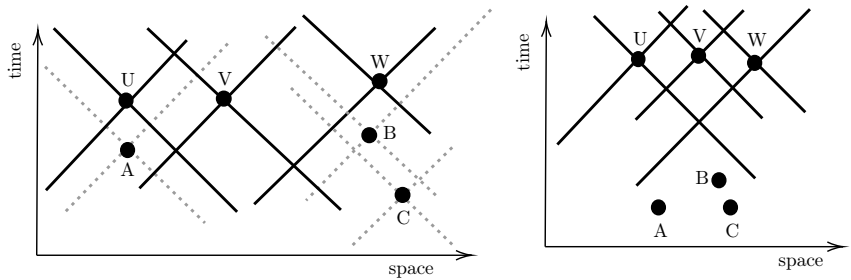
  \centering 
\includestandalone[width=0.8\textwidth]{tikzfigures/tikz_GHZ_leegwater_BafterC}  \caption{Two lightcone structures in which a paradox occurs in the GHZ--FR protocol, but to which the Relativistic Quantum Mechanics argument of \cite{leegwater2022greenberger} cannot be applied to obtain a paradox, as $B$ occurs after $C$ and thus $B,C$ never occur simultaneously in the same inertial frame.} \label{fig:GHZ_FR_leegwater_BafterC} 
\end{figure}

\newpage

\nocite{apsrev42Control}

\bibliography{refs}

\end{document}